\author{} 

\title{} 
\date{\today} 

\documentclass[twoside,11pt]{article}

\usepackage{jmlr2e}

\usepackage[left=30mm,top=30mm,right=30mm,bottom=30mm]{geometry}
\usepackage{etoolbox} 
\usepackage{booktabs}
\usepackage[table,xcdraw]{xcolor}
\usepackage[usestackEOL]{stackengine}
\usepackage[T1]{fontenc}
\usepackage[utf8]{inputenc}
\usepackage{bm}
\usepackage{graphicx}
\usepackage{subcaption}
\usepackage{amsmath}
\usepackage{amsfonts}
\usepackage{mathtools}
\usepackage{xcolor}
\usepackage{float}
\usepackage{hyperref}
\usepackage[capitalise]{cleveref}
\usepackage{enumitem,kantlipsum}
\usepackage{amssymb}
\usepackage{amsbsy}
\usepackage{amsthm}
\usepackage{bbm}
\usepackage{pifont}
\usepackage{tcolorbox}

 \usepackage{multirow}
\usepackage{algorithm}
\usepackage{algpseudocode}
\usepackage{listings}
\usepackage{dirtytalk}
\usepackage{graphicx}

\usepackage{chngcntr}
\usepackage{apptools}
\AtAppendix{\counterwithin{lemma}{section}}

\usepackage{mathrsfs}
\usepackage{wrapfig}

\newcommand{\E}{\mathbb{E}}

\newcommand{\Cov}{\mathrm{Cov}}

\newtheorem{definition}{Definition}

\newtheorem{example}{Example}

\newtheorem*{assumption*}{Assumption}

\newtheorem*{theoremcustommmm}{Theorem 4+5}
\newtheorem{innercustomthm}{Theorem}
\newenvironment{customthm}[1]
  {\renewcommand\theinnercustomthm{#1}\innercustomthm}
  {\endinnercustomthm}

\newtheorem{innercustomlem}{Lemma}
\newenvironment{customlem}[1]
  {\renewcommand\theinnercustomlem{#1}\innercustomlem}
  {\endinnercustomlem}

  \newtheorem{theorem}{Theorem}
\newtheorem{lemma}{Lemma}

\newcommand{\indep}{\perp \!\!\! \perp}

\hypersetup{
    colorlinks,
    linkcolor={black},
    citecolor={blue!50!black},
    urlcolor={blue!80!black}
}

\graphicspath{{figures/}}

\begin{document}

\title{Cross-World Assumption and Refining Prediction Intervals for Individual Treatment Effects}

\author{\name Juraj Bodik \email juraj.bodik@unil.ch \\
       \addr Department of Statistics, University of California, Berkeley, USA \\Department of Operations, HEC, University of Lausanne, Switzerland
       \AND
       \name Yaxuan Huang \email yaxuan\_huang@berkeley.edu \\
       \addr  Department of Statistics, University of California, Berkeley, USA
        \AND
       \name Bin Yu \email binyu@berkeley.edu \\
       \addr  Department of Statistics, University of California, Berkeley, USA \\ Department of Electrical Engineering and Computer Science, UC Berkeley, USA}

\maketitle

\begin{abstract}

While average treatment effects (ATE) and conditional average treatment effects (CATE) provide valuable population- and subgroup-level summaries, they fail to capture uncertainty at the individual level. For high-stakes decision-making, individual treatment effect (ITE) estimates must be accompanied by valid prediction intervals that reflect heterogeneity and unit-specific uncertainty. However, the fundamental unidentifiability of ITEs limits the ability to derive precise and reliable individual-level uncertainty estimates. To address this challenge, we investigate the role of a cross-world correlation parameter, \( \rho(x) = \text{cor}\big(Y(1), Y(0) \mid X = x\big) \), which describes the dependence between potential outcomes, given covariates, in the Neyman-Rubin super-population model with i.i.d. units.  Although \( \rho \) is fundamentally unidentifiable, we argue that in most real-world applications, it is possible to impose reasonable and interpretable bounds informed by domain-expert knowledge. Given $\rho$, we design prediction intervals for ITE, achieving more stable and accurate coverage with substantially shorter widths; often less than 1/3 of those from competing methods. The resulting intervals satisfy coverage guarantees $P\big(Y(1) - Y(0) \in C_{ITE}(X)\big) \geq 1 - \alpha$ and are asymptotically optimal under Gaussian assumptions. We provide strong theoretical and empirical arguments that cross-world assumptions can make individual uncertainty quantification both practically informative and statistically valid.

\end{abstract}
  

\pagenumbering{arabic}

\section{Introduction}

The average treatment effect (ATE) is a common standard for evaluating causal effects, offering a population-level summary of the impact of an intervention \citep{Imbens_Rubin_book}. In many decision-making settings, relying solely on the ATE is insufficient, as it fails to capture the heterogeneity in individual responses to treatment. In personalized medicine, policy evaluation, and economics, understanding individual treatment effects (ITE) is crucial for tailoring interventions to a specific individual \citep{athey2016recursive, robins2000marginal}.  Nevertheless, estimating ITE comes with significant challenges, particularly in constructing reliable prediction intervals that quantify uncertainty.  

A fundamental challenge in estimating ITE is its inherent non-identifiability. Since we can never observe both potential outcomes $Y(0), Y(1)$ for the same individual, the true ITE remains unknown. This limitation is deeply connected to Pearl’s third ladder of causation \citep{TheBookOfWhy}, which deals with counterfactual reasoning. While association (first ladder) and intervention (second ladder) can be addressed through statistical and experimental methods, counterfactual inference requires untestable assumptions about the unobserved hypothetical ``What if'' world. 

A natural way to mitigate the challenges of estimating ITE is to consider the Conditional Average Treatment Effect (CATE). Unlike ITE, which remains fundamentally unobservable, CATE can be consistently estimated using statistical and machine learning methods under standard causal assumptions \citep{Hahn1998, Wager_2024_Book_Causal}. While CATE provides valuable insights into treatment heterogeneity, it does not offer direct uncertainty quantification for individual effects; only for group-level averages. In many applications, such as personalized medicine or policy evaluation, decision-makers require prediction intervals rather than just point estimates with confidence intervals, ensuring reliable risk assessment for each individual. 

Several recent works have proposed methods for constructing prediction intervals for ITEs \citep{Lei_2021_Conformal_Inference, Kivaranovic2020, Alaa_2023_Conformal_Meta_learners, jonkers2024conformalconvolutionmontecarlo}. However, these approaches either result in extremely wide intervals with limited practical utility, or lack coverage guarantees due to the fundamental unidentifiability of ITEs.  We aim to bridge this gap by producing narrow prediction intervals for individual treatment effects while maintaining rigorous coverage guarantees. To achieve this, we rely on a new assumption, which can be specified or bounded using expert domain knowledge.

\begin{definition}[Cross-world assumption]
In the Neyman–Rubin super-population model with i.i.d. units, the dependence structure (conditional correlation) between the potential outcomes $Y(1), Y(0)$, conditioned on the observed covariates \( X \), is defined as:
\begin{equation}
\label{rho_definition}
    \rho(x) = \operatorname{cor}\big(Y(1), Y(0) \mid X=x\big).
\end{equation}
We refer to an assumption about $\rho$ as a cross-world assumption. 
\end{definition}
While the term \textit{``cross-world assumption''} is newly introduced in this work, related ideas have appeared in prior literature (see related literature in Section~\ref{subsection_main_cross_world3}), albeit without being explicitly named. These assumptions are typically embedded within additive structural equation models, in which the dependence between potential outcomes is captured by the correlation $\rho$ between the exogenous noise terms $\varepsilon_0$ and $\varepsilon_1$ in the model:
\begin{equation*}
    Y_i(0) = \mu_0(X_i)+ \varepsilon_0,  \quad  Y_i(1) = \mu_1(X_i)+ \varepsilon_1,\,\,\,\,\,\,\text{ where }\operatorname{cor}(\varepsilon_1,\varepsilon_0) = \rho(X_i). 
\end{equation*}
Although \( \rho \) is unidentifiable, we argue that incorporating expert domain knowledge to postulate at least plausible bounds on \( \rho \) is both reasonable and aligned with how humans make judgments in practice. Observing one potential outcome is typically highly informative about the unobserved counterfactual, often conveying information beyond what is captured by measured covariates \citep{Chernozhukov2023Toward, wu2025counterfactual}. Despite it being untestable, humans make use of counterfactual statements all the time, which is  necessary for intelligent behavior  \citep{lewis2020counterfactuals}. 

Consider the following example: a randomized clinical trial for a new antihypertensive drug. The potential outcomes
$Y_i(1)$ and $Y_i(0)$ denote the patient’s systolic blood pressure one year later under
treatment and control. Observed covariates $X_i$ include age, sex, baseline blood
pressure, and basic health indicators. Even after conditioning on $X_i$, important unobserved
physiological factors remain, such as genetic predisposition, vascular stiffness, and long-term
lifestyle habits. These latent traits affect blood pressure under both treatment regimes in a similar manner. Consequently, individuals with high $Y_i(0)$ tend also to have high $Y_i(1)$. As a result, the cross-world correlation $\rho(x)$ is naturally positive and often plausibly large. Assuming $\rho(x) > 0$, or even $\rho(x) > 0.5$, is therefore reasonable in similar settings,
as the intervention primarily shifts outcomes without substantially altering the relative
ordering of individuals.

In Section~\ref{section_2} we delve deeper into the discussion about the cross-world assumption. Nevertheless, our main goal is to construct the smallest possible prediction intervals $C_{ITE}$ for individual treatment effect satisfying for given $\alpha\in(0,1)$:
\begin{equation}
    \label{marginal_coverage_definition}
\mathbb{P}\big(Y_{}(1) -  Y_{}(0) \in C_{ITE}(X_{})\big) \geq 1 - \alpha,
\end{equation}
or more ambitiously,
\vspace{-0.2cm}
\begin{equation}
    \label{def_cond_coverage}
    \mathbb{P}\big(Y(1) -  Y(0) \in C_{ITE}(X) \mid X=x\big) \geq 1-\alpha.
\end{equation}

\textbf{Our contributions} represent, to our knowledge, the first to establish rigorous finite-sample and asymptotic guarantees for individual-level prediction intervals under cross-world assumption. Specifically:
\begin{itemize}
\item We clarify the role of the cross-world correlation $\rho$ in ITE uncertainty quantification, and provide practical tools for reasoning about, bounding, and \textit{estimating} $\rho$ using auxiliary covariates and domain knowledge. This includes a novel way of leveraging covariates to sharpen prediction intervals. 
\item We propose a novel class of prediction intervals, denoted $CW(\rho)$ (Cross-World Intervals), which incorporate cross-world assumption on $\rho$. 
\item When $\rho$ is known, we prove that $CW(\rho)$ intervals are, under some assumptions, asymptotically optimal in a sense that they are the smallest intervals that satisfy \eqref{def_cond_coverage}. Even when $\rho$ is only bounded, or sample size is small, they also enjoy strong finite-sample coverage guarantees \eqref{marginal_coverage_definition}.
\item We provide empirical evaluations that the $CW(\rho)$ intervals (in particular, the $CW^{+CI}(\rho)$ variant) consistently maintain stable and valid coverage while achieving significantly narrower widths, often less
than 1/3 of those based on existing methods, even when $\rho$ is slightly misspecified or the noise is non-Gaussian.
\end{itemize}

\textbf{The rest of the paper is organized as follows.} Section~\ref{section_2} provides a detailed discussion of the cross-world assumption and the related literature. We introduce $CW(\rho)$ and $CW^{+CI}(\rho)$ prediction intervals in Section~\ref{section4}, and provide theoretical guarantees, including asymptotic optimality and finite-sample coverage when $\rho$ is known or bounded, in Section~\ref{section5}. Empirical evaluations on simulated and real data are presented in Section~\ref{section_simulations}. Additional theoretical results are provided in Appendix~\ref{section_appendix_theoretical_guarantees}. Further empirical results appear in Appendix~\ref{appendix_simulations}, and all formal proofs are deferred to Appendix~\ref{appendix_proofs}.

\section{Cross-World Assumption}
\label{section_2}

\subsection{Causal inference: notation and preliminaries}
\label{section_preliminaries}

Throughout this paper, we adopt the Neyman–Rubin potential outcome framework \citep{Neyman1923, Rubin1974} and assume a super-population from which a realization of $n$ independent random variables is given as the training data. For each of the \( n \) observed units, let \( Y_i(1) \) and \( Y_i(0) \) denote the potential outcomes for unit \( i \) under treatment and control, respectively. The ITE is defined as \(Y_i(1) - Y_i(0)\).
The fundamental issue of causal inference is that the ITE is unobservable: for each unit, we can only observe one of the two potential outcomes. Let \( T_i \in \{0,1\} \) indicate treatment assignment, where \( T_i = 1 \) if unit \( i \) receives treatment and \( T_i = 0 \) otherwise. The observed outcome is then given by  
\(
Y_i = T_i Y_i(1) + (1 - T_i) Y_i(0).
\) 
In addition to treatment and outcome, we assume access to a vector of pre-treatment covariates \( X_i \in \mathcal{X}\subseteq\mathbb{R}^d \). 
We consider $\big(Y_i(1), Y_i(0), T_i, X_i\big) \overset{\text{i.i.d.}}{\sim} \big(Y(1), Y(0), T, X\big)$ for a generic random vector. Additionally, we define ATE as $\mathbb{E}[Y(1) - Y(0)]$ and CATE as 
$$
\tau(\textbf{x}) = \mu_1(\textbf{x})-\mu_0(\textbf{x}), \,\,\,\,\,\text{where}\,\,\,\,\mu_t(\textbf{x}) =  \mathbb{E}[Y(t) |  \textbf{X} = \textbf{x}].
$$
Throughout the paper, we assume the standard causal assumptions of \textbf{strong ignorability} and \textbf{overlap}: \(
(Y(1), Y(0)) \perp\!\!\!\perp T \mid X
\) and $0 < \pi(\textbf{x}) < 1$ for all $\textbf{x}\in \mathcal{X}$, where $\pi(\textbf{x})=\mathbb{P}(T = 1 \mid \textbf{X}=\textbf{x})$ is usually referred to as a propensity score. This implies that the treatment assignment \( T \) does not depend on the potential outcomes after conditioning on the covariates \( X \), and each unit has a non-trivial probability of being assigned to either treatment or control. These assumptions allow us to identify CATE using $\mu_t(\textbf{x}) =  \mathbb{E}[Y |  T=t, \textbf{X} = \textbf{x}]$ \citep{Imbens_Rubin_book}.

Several approaches have been proposed for estimating CATE, including flexible meta-learners such as T-, S-, X-, and DR-learners \citep{Kunzel_2019_Metalearners_Estimating, kennedy_2023_optimal_doubly_robust}, or tree-based methods like causal forests \citep{athey_2016_recursive_partitioning, Wager_2018_Estimation_Inference}. CATE estimators are often compared using validation-based risk measures \citep{schuler2018comparisonmethodsmodelselection, Dwivedi_2020_Stable_Discovery}, and ensemble methods have been developed to combine strengths across models and reduce estimation risk \citep{nie2021quasioracle, han2022ensemblemethodestimatingindividualized}.

Two types of ``no treatment effect'' null hypotheses are commonly considered in causal inference. Neyman’s null hypothesis is a population-level statement that $ATE=0$ \citep{Neyman1923}. Fisher’s null hypothesis is a sharp null stating that the treatment has no effect on any individual $ITE = 0$ \citep{Fisher1935}. The sharp null is a stronger assumption and indeed an entire population can exhibit zero average effect despite heterogeneous individual effects that cancel out \citep{Ding2017Paradox}.

\subsection{Cross-world assumption: connection with existing literature}
\label{subsection_main_cross_world3}

While ignorability assumption \(
\big(Y(1), Y(0)\big) \perp\!\!\!\perp T \mid X
\) is extremely prominent in the causal literature, assumptions on the joint distribution of \((Y(1), Y(0))  \mid X\) are rare. Although we observe the marginals $Y(1) \mid X$ and $Y(0) \mid X$, their dependence is unidentifiable because only one of $Y(0)$ and $Y(1)$ is observed.

Due to the linearity of expectation, we can identify CATE without any cross-world assumption, since always $\mathbb{E}[Y(1) - Y(0) \mid X] = \mathbb{E}[Y(1) \mid X]-\mathbb{E}[Y(0) \mid X]$. \cite{Rubin1990} and \citep[subsection 3.2]{Rubin_rho_for_CATE_as_sensitivity_parameter}  considered $\rho$ as a sensitivity parameter and showed that it plays no role for an inference of CATE.  However, this does not extend to functionals like quantiles or variances. \citet{Gadbury2000} raised this issue in obesity research. In the Noble lecture and subsequent work, \cite{heckman1997making, heckman2001micro, carneiro2003estimating} use parametric factor models to estimate the joint counterfactual distribution. \cite{cai2022relatepotentialoutcomesestimating} used $\rho$ for estimating missing counterfactual in linear Gaussian model.  

The parameter $\rho$ naturally arises in discussions of ITE variance \citep{Rubin1990}. \citet{ding2019decomposing} showed that idiosyncratic variance is inherently cross-world and unidentifiable without joint assumptions. Another line of work \citep{zhang2025individualtreatmenteffectprediction, zhang2025boundsdistributionsumrandom} seeks to tighten ITE prediction intervals by leveraging Fréchet–Hoeffding-type copula bounds. Recent work, such as \citet{wu2024quantifying}, incorporates $\rho$ to estimate the Fraction Negatively Affected (FNA) for binary outcomes. Our work extends this perspective to continuous outcomes.

\cite{Chernozhukov2023Toward}, in a commentary on \cite{Jin2023Sensitivity}, discussed how $\rho$ can be used to define conservative conformal prediction intervals for ITE. They advocated varying $\rho$ to explore sensitivity and robustness, and noted that assuming a Gaussian copula could mitigate over-conservativeness under certain regularity conditions. We build on this, focusing on the Gaussian copula due to (i) its plausibility for continuous outcomes, (ii) since 
$Y(1)$ and $Y(0)$ are never observed together, more flexible models cannot be effectively used, (iii) interpretability of $\rho$, and (iv) robust empirical performance across settings.

In the Structural Causal Models (SCM, \cite{Pearl_book, Peters2017, bodik2024structuralrestrictionslocalcausal}), we assume a data-generating process of the form $Y(T) = f_Y(T, X, \varepsilon)$, $\varepsilon\indep T, X$. In an additive case $Y = f_Y(T,X)+\varepsilon$, we automatically have  $\operatorname{cor}\big(Y(1), Y(0) \mid X\big) = 1$.  However, this is no longer true in more complicated data-generating processes; consider, for example $Y = T\varepsilon + (1-T)\varepsilon^2$, with $\varepsilon\sim N(0,1)$.
In such a case, knowledge of $Y(0)$ brings only limited information about $Y(1)$ (large $Y(0)$ can mean either large $Y(1)$ or large $-Y(1)$) and $\operatorname{cor}\big(Y(1), Y(0) \mid X\big) = 0$. The previous example highlights the fact that in very simple data-generating processes where the external factors ($\varepsilon$) have an additive form, we typically obtain $\operatorname{cor}\big(Y(1), Y(0) \mid X\big) \approx 1$. However, in real-world scenarios, $\varepsilon$ is complex and affects the actual data-generating processes in a non-additive way, and $cor(Y(1), Y(0)\mid X)$ can be potentially very different from $1$. 

 \begin{figure}[h!]
\centering
\includegraphics[width=\textwidth]{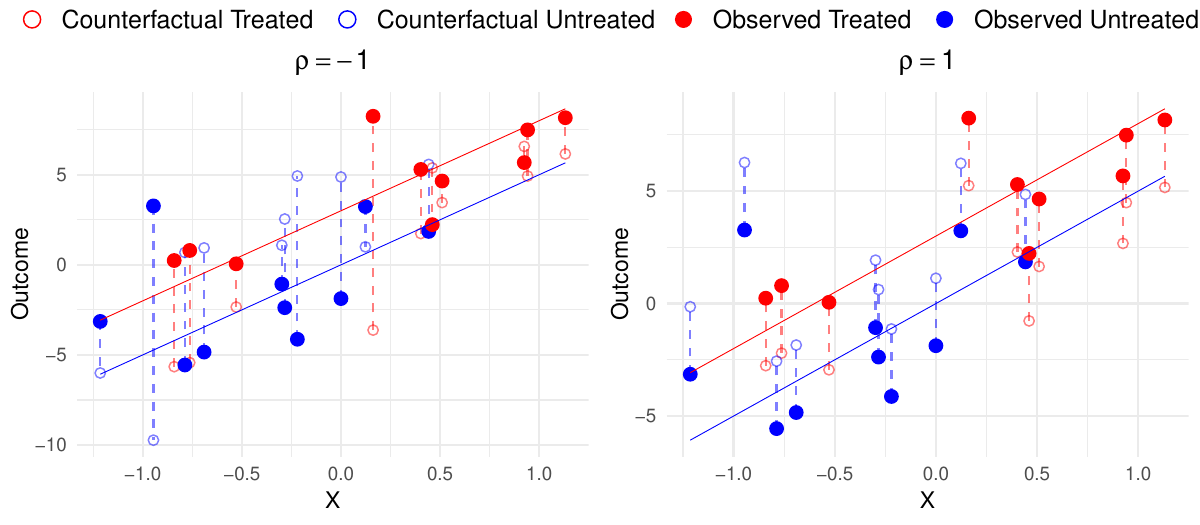}
\caption{Difference between perfectly negatively dependent ($\rho = -1$) and perfectly positively dependent ($\rho = 1$) potential outcomes. In both cases, the observed variables are the same. However, much larger prediction intervals for ITEs are needed when $\rho = -1$, while $\rho = 1$ leads to the narrowest possible intervals. The solid lines represent conditional expected values. }
\label{Rho_figure}
\end{figure}

\subsection{Plausible values of \texorpdfstring{$\rho$}{rho} in practice: domain knowledge, auxiliary data, and sensitivity analysis}

The cross-world correlation $\rho(x)$ is not identifiable from standard observed data, so reasoning about it must rely on domain knowledge, auxiliary data, or both. We therefore recommend sensitivity analysis over a plausible range of $\rho$ values, while noting that auxiliary covariates or domain knowledge may sometimes justify narrower ranges or informative lower bounds. In the following we present practical and transparent ways to assess plausible values of $\rho(x)$, illustrated with data examples, and offers a guideline for sensitivity analysis.

\subsubsection{Domain knowledge reasoning via simplified latent-structure}
\label{subsubsection2.21}

A useful starting point is to consider simple settings that clarify what kinds of cross-world dependence are scientifically plausible. We do not encourage such examples for estimation; rather, they provide an interpretable way to think about the sign and magnitude of dependence between $Y(1)$ and $Y(0)$.

A natural benchmark is the constant-treatment-effect setting: if
$Y(1)=Y(0)+f(X)$
for some function $f$, then $\rho(x)=1$ for all $x$. This includes Fisher's sharp null \citep{Fisher1935} as the special case $f(\cdot)\equiv 0$. While constant treatment effects are typically too restrictive to be taken literally, this suggests that when we expect the ITE $Y_i(1)-Y_i(0)$ to be small relative to the baseline heterogeneity in outcomes, assuming a high cross-world correlation is plausible and natural.

A related and often useful way to reason about $\rho(x)$ is through a simplified latent-structure model. Consider the decomposition
\begin{equation}\label{eq_hidden_covariate}
    Y(1) = \mu_1({X}) + \underbrace{H + \tilde\varepsilon_1}_{ \varepsilon_1},\quad
Y(0) =  \mu_0({X})  + \underbrace{H + \tilde\varepsilon_0}_{\varepsilon_0},\quad
\tilde\varepsilon_0 \indep \tilde\varepsilon_1,
\end{equation}
where $H\indep ({X},T)$ represents an unobserved factor that affects both potential outcomes in a similar way, while $\tilde\varepsilon_1$ and $\tilde\varepsilon_0$ capture treatment-specific residual variation. Under the additional assumptions that $H \indep \tilde\varepsilon_1,\tilde\varepsilon_0$, the cross-world parameter $\rho(x)$ satisfies
\[
\rho(x)
=
\frac{\operatorname{Var}(H)}
{\sqrt{\operatorname{Var}(H)+\operatorname{Var}(\tilde\varepsilon_1)}
 \sqrt{\operatorname{Var}(H)+\operatorname{Var}(\tilde\varepsilon_0)}}
\ge 0.
\]
This example captures a simple and useful idea: when a substantial part of unexplained variation is driven by latent components that affect both treatment states similarly, the two potential outcomes tend to be positively dependent. It is therefore natural to ask \textit{what fraction of the residual outcome variation is plausibly due to common latent drivers}, such as underlying health status, baseline ability, or stable individual features. The answer provides a meaningful starting point for choosing plausible values of $\rho(x)$.

The decomposition in \eqref{eq_hidden_covariate} is only a simplified tool for domain-based reasoning. In practice, latent factors may enter nonadditively and interact with treatment, which can weaken or even reverse the dependence between $Y(1)$ and $Y(0)$. We do not claim that large positive cross-world dependence is universal, rather that it is a natural default in settings where stable latent traits affect both potential outcomes similarly or rank preservation assumption is reasonable. The next subsections support this view through auxiliary-data lower bounds, benchmark datasets with near-paired outcomes, and motivate sensitivity analysis over plausible values of $\rho$. 

\subsubsection{Estimating lower bounds on $\rho$ via auxiliary data}
\label{section_rho_with_Z}

In many applications, one observes a collection of \textbf{auxiliary covariates} $Z$ in some subset of data, in addition to the covariates $X$ used for ITE modeling. For example, in a clinical setting, a subset of patients may have rich auxiliary measurements such as detailed medical history, lab panels, or genetic markers, whereas for most new incoming patients only a small set of routine covariates $X$ (e.g., age, sex, and a few basic vitals) is available at decision time. The presence of $Z$ allows us to construct informative, data-driven \emph{bounds} on $\rho(x)$.

By the law of total covariance,\vspace{-0.5cm}
\begin{align}
\mathrm{Cov}\{Y(0),Y(1)\mid X=x\}
&=
\overbrace{\mathrm{Cov}\{\mu_0(x,Z),\mu_1(x,Z)\mid X=x\}}^{=:A(x)}
\notag\\
&\quad+
\underbrace{\mathbb{E}\!\left[\mathrm{Cov}\{Y(0),Y(1)\mid X=x,Z\}\,\middle|\, X=x\right]}_{=:B(x)},
\label{eq:cov_decomp_lower}
\end{align}
where $\mu_t(x,z):=\mathbb{E}[Y(t)\mid X=x,Z=z]$ for $t\in\{0,1\}$.

The term $A(x)$ is identified from the observed data under standard causal assumptions, and captures the portion of cross-world dependence explained through $Z$ above the baseline covariates $X = x$.  The term $B(x)$ is non-identifiable, but it can be bounded using Cauchy--Schwarz inequality, which gives us an identifiable lower bound for $\rho(x)$:
\begin{lemma}\label{lemma_rho_L}
   Assume that the relevant conditional second moments exist. Denote  $\sigma_t(x):=\mathrm{Var}\{Y(t)\mid X=x\}$, and $\sigma_t(x,z):=\mathrm{Var}\{Y(t)\mid X=x,Z=z\}$, $t=0,1$. Then $$B(x)\geq - B_L(x), \qquad \text{where }\qquad B_L(x):=\mathbb{E}\!\left[\sqrt{\sigma_0(x,Z)\,\sigma_1(x,Z)}\,\middle|\, X=x\right],$$ and consequently \begin{align}
\rho(x)\ \ge\ \tilde\rho_L(x)
:= \max\!\left\{
-1,\,
\frac{A(x)-B_L(x)}{\sqrt{\sigma_0(x)\sigma_1(x)}}
\right\}.
\label{eq:rho_lower}
\end{align}
 Additionally, if we assume $\mathrm{Cov}\{Y(0),Y(1)\mid X,Z\}\ge 0$, then the bound tightens to
\begin{equation}\label{eq_rho_L}\rho(x)\ge \rho_L(x):= A(x)/ \sqrt{\sigma_0(x)\sigma_1(x)}.\end{equation}
\end{lemma}

\textbf{Direct estimation of} $\rho_L(x)$. Constructing estimators $\widehat{\rho}_L(x)$ and
$\widehat{\tilde\rho}_L(x)$ is straightforward; see Appendix~\hyperref[first_equation_in_appendix_A]{A}
for details, consistency guarantees, and simulations. In practice, performance hinges on two
factors: (i) the auxiliary covariates $Z$ must explain a non-negligible portion of the outcome
variation, otherwise the resulting bound is overly conservative; and (ii) the nuisance regressions
$\widehat{\mu}_0(x,z)$ and $\widehat{\mu}_1(x,z)$ must be estimated accurately, which can be difficult
in small samples, especially when $Z$ is high-dimensional. When either condition fails, the resulting
lower bound (or a lower confidence limit based on $\widehat{\rho}_L$) may become too conservative.

\paragraph{Imputation-based estimation of $\rho$. }
In Appendix~\hyperref[first_equation_in_appendix_A]{A}, we complement the bound-based approach with the imputation-based estimator $\widehat\rho_{\mathrm{imp}}(x)$, which replaces the missing counterfactual outcome $\widehat Y_i(1-T_i)$ for each unit by a prediction based on $(X,Z)$ (done e.g. by matching). The cross-world correlation is then estimated as the conditional correlation between $\hat{Y}(0),\hat{Y}(1)$ of these pseudo-pairs given $X=x$ (implemented, e.g., by binning or kernel smoothing). While this estimator is generally not guaranteed to produce a lower bound on $\rho(x)$ without additional assumptions on the imputation error, it can often be a practically useful choice. 

\subsubsection{Values of $\rho$ in many real-world datasets.}

We illustrate these ideas on several widely used benchmark datasets; full details are deferred to Appendix~\ref{appendix_simulations_rho}. First, the \textbf{Twins dataset} \citep{Twins_dataset} provides a quasi-observational setting in which paired potential outcomes can be approximated by comparing outcomes within twin pairs. For strata defined by covariates $(X_i,X_j)=(x_i,x_j)$, we compute the empirical Pearson correlation between $Y(1)$ and $Y(0)$ within each stratum. For typical covariate pairs $(i,j)$, we find (see Appendix~\ref{appendix_twins_rho} for details)
\[
\rho(x_i,x_j)
=
\operatorname{cor}\bigl(Y(1),Y(0)\mid X_i=x_i, X_j=x_j\bigr)
\approx 0.6 \pm 0.1.
\]
In a more realistic setting in which only the factual outcome $Y=Y(T)$ is observed, we applying the auxiliary-data approach from Section~\ref{section_rho_with_Z}, with $(X_i,X_j)$ treated as observed covariates and the remaining variables used as auxiliary covariates. This yields typical values around $\widehat\rho_L(x_i,x_j)\approx 0.4$, while the fully assumption-free bound gives $\widehat{\tilde\rho}_L(x_i,x_j)\approx -0.1$.

Second, we consider \textbf{30 cross-over studies} \citep{Senn2002}, in which each unit receives both treatments $T=0$ and $T=1$ in different periods. Under the standard assumption that period and carryover effects are negligible, the two observed outcomes can be viewed as a close proxy for the two potential outcomes of the same unit, allowing a direct paired estimate of the cross-world correlation. The dataset compilation of \cite{Schutz2020ReplicateBE} contains 30 such studies; Appendix~\ref{appendix_Cross_over_studies} reports the estimated values of $\rho$ for each. Overall, we find $\rho>0.5$ in 25 out of 30 datasets, and $\rho>0.33$ in 29 out of 30 datasets, with lowest $\rho\approx 0.17\pm 0.2$. 

Finally, we consider the \textbf{IHDP dataset} \citep{Bayes_Hill}. Although the original data-generating mechanism uses independent noise terms for the two potential outcomes, conditioning only on a reduced covariate set can still induce positive cross-world correlation, with typical values around $\rho(x)\approx 0.35$. Taken together, these examples suggest that positive cross-world dependence is common in benchmark settings, and that richer covariates or auxiliary variables can substantially sharpen practically relevant lower bounds on $\rho(x)$.

\subsubsection{Practitioner recipe}

In practice, we recommend treating $\rho(x)$ as a sensitivity parameter and combining sensitivity analysis with any available domain knowledge or auxiliary-data-based information.

\begin{enumerate}
\item \textbf{Start with some plausible range.} Choose a range of $\rho$ values that is scientifically reasonable for the application. In many settings, it is natural to focus on nonnegative values when the two potential outcomes are expected to share substantial common drivers; in other settings, a wider range may be appropriate. When a conservative analysis is desired, one may also report intervals corresponding to smaller values of $\rho$, including worst-case choices such as $\rho=-1$, or, when a constant treatment effect is expected, one may report intervals corresponding to $\rho=1$.
\item \textbf{Reason with domain knowledge to calibrate the range.} Simple latent-structure considerations, such as those in Section~\ref{subsubsection2.21}, can help assess whether high positive dependence is plausible, whether values near zero are more appropriate, or whether negative dependence should be considered.
\item \textbf{Use auxiliary covariates when available.} When richer covariates $Z$ are available and nuisance estimation is feasible, compute $\widehat{\tilde\rho}_L(x)$ and, when scientifically appropriate, $\widehat\rho_L(x)$. These provide data-driven bounds that can help narrow the sensitivity range.
\item \textbf{Report intervals across representative values of $\rho$.} Rather than selecting a single default value, report intervals for several representative choices of $\rho$ within the chosen range. This shows directly how the assumed cross-world dependence affects interval width and helps readers interpret the substantive consequences of different assumptions.
\end{enumerate}

\color{black}

\section{Constructing Prediction Intervals for ITE under Cross-World Assumption }
\label{section4}

Our objective is to construct the smallest possible prediction set \( C_{ITE} \) that satisfies marginal coverage \eqref{marginal_coverage_definition}, or ideally, conditional coverage \eqref{def_cond_coverage}. If the individual treatment effects \( Y_i(1) - Y_i(0) \) were observable, we can construct the prediction intervals as in a classical regression setting, as described in Section~\ref{section_UQ_preliminaries}. 

\subsection{Uncertainty quantification and prediction intervals in classical regression}
\label{section_UQ_preliminaries}
In a classical regression setting, we observe a random sample \( (X_i, Y_i) \overset{i.i.d}{\sim} P_X \times P_{Y \mid X} \), for \( i = 1, \dots, n \), and aim to construct a prediction set \( C(X) \subseteq \text{supp}(Y) \) satisfying marginal coverage or conditional coverage, defined respectively as: 
\begin{equation*}
 \mathbb{P}\big(Y_{n+1} \in C(X_{n+1})\big) \geq 1 - \alpha, \quad \text{ or }    \quad \mathbb{P}\big(Y_{n+1} \in C(X_{n+1}) \mid X_{n+1} = x\big) \geq 1 - \alpha,
\end{equation*}
where \( \alpha \in (0,1) \) (e.g., \( \alpha = 0.1 \)) and \( x\in\text{supp}(X) \) lies in the support of the covariate distribution. While conditional coverage is more informative, marginal coverage can be achieved in a finite-sample, model-free setting. In contrast, \cite{barber2020limitsdistributionfreeconditionalpredictive} show that achieving conditional coverage generally requires strong model assumptions or relies on asymptotic guarantees.

Conformal inference \citep{vovk_2005_algorithmic_learning, vovk_2009_online_predictive, vovk_2012_conditional_validity, vovk_2015_cross_conformal, angelopoulos2024theoreticalfoundationsconformalprediction} provides a framework for producing valid prediction sets with guaranteed finite-sample marginal coverage without relying on strong assumptions about the joint distribution of the data. 
In order to construct $C$, data $\mathcal{D} = \{(X_i, Y_i), i=1, \dots, n\}$ are split into train and calibration subsets $\mathcal{D}_{train}, \mathcal{D}_{calib}$. On the training data, a first estimate of $\tilde{C}$ is constructed, and then we use data from $\mathcal{D}_{calib}$ to calibrate $\tilde{C}$ such that marginal coverage is satisfied. A well-known instance is Conformalized Quantile Regression (CQR, \cite{Romano_2019_Conformalized_Quantile}), which leverages conditional quantile estimates to form $\tilde{C}$, yielding asymptotically valid conditional coverage.

The strong marginal validity guarantees of conformal methods rely on the exchangeability assumption.  While this assumption is typically reasonable in randomized control trials, it is often violated in observational studies due to covariate shift between treated and untreated individuals. If we can produce conditionally valid prediction intervals, the covariate shift is not an issue since conditional coverage implies marginal coverage for any distribution of $X$. However, guarantees of conditional validity are very scarce \citep{gibbs_conditional_conformal}. 

To tackle this challenge, \cite{Tibshirani_2019b_Conformal_Prediction} extended the traditional conformal prediction framework by introducing weighted conformal prediction. This method guarantees correct marginal coverage when a known likelihood ratio between \(X\mid T=1\) and \(X\mid T=0\) is available. Subsequently, \cite{Lei_2021_Conformal_Inference} demonstrated that even when this likelihood ratio (or equivalently, the propensity score \(\pi(x)\)) is estimated, the framework still achieves asymptotically valid marginal coverage, accompanied by strong empirical performance.

\subsection{\texorpdfstring{$CW(\rho)$}{D} prediction intervals }

While a construction of a prediction set \( C_{ITE} \) would be straightforward if the individual treatment effects \( Y_i(1) - Y_i(0) \) were observable, we must work with factual outcomes alone. As a starting point, we construct prediction intervals $C_1$ and $C_0$ for the treated and control groups separately:  
\begin{equation}
    \label{separate_def}
    \mathbb{P}\big(Y(1) \in C_{1}(X)\big) \geq 1 - \alpha, \quad \mathbb{P}\big(Y(0) \in C_{0}(X)\big) \geq 1 - \alpha,
\end{equation}
which can be guaranteed if any conformal method is used. Suppose their form: 
\begin{equation}
\begin{aligned}
        C_1(x) &= [\hat{\mu}_1(x) -l_1(x), \hat{\mu}_1(x) +u_1(x) ], \\
        C_{0}(x) &= [\hat{\mu}_0(x) -l_0(x), \hat{\mu}_0(x) +u_0(x) ],
\end{aligned}
\label{C_1_and_C_0_definitions}
\end{equation}
where $\hat{\mu}_t$ is an estimate of $\mu_t$ and $l_t, u_t\geq 0$ are the (lower and upper) widths of prediction intervals for $Y(t), t=0,1$. We combine \( C_1(X) \) and \( C_0(X) \) into a prediction set \( C_{\text{ITE}}(X) \) for the individual treatment effect as follows. 

\begin{definition}[$CW(\rho)$ intervals]
Let $\rho\in [-1, 1]$. Consider the individual prediction intervals for $Y(1), Y(0)$ in the form \eqref{C_1_and_C_0_definitions}. Denote an estimate of CATE as $ \hat{\tau}(x)$ (possibly $\hat{\tau}(x) = \hat{\mu}_1(x) - \hat{\mu}_0(x)$). We define $CW(\rho)$ intervals as follows: 
\begin{equation}
\begin{aligned}   \label{D_rho_intervals_definition}
   CW(\rho)\text{ intervals:}\quad\quad &C_{ITE}(x)=\left[ \hat{\tau}(x) - D_\rho\big(l_1(x), u_0(x)\big), \quad 
    \hat{\tau}(x) + D_\rho\big(l_0(x), u_1(x)\big) \right], \\
    & \text{where}\footnotemark\,\,\,\, D_\rho(x,y) = \sqrt{x^2 + y^2 - 2\rho xy}.
\end{aligned}
\end{equation}
\footnotetext{Given that we did not find a similar concept in the literature, we refer to $D_\rho$ as the Correlation-Adjusted Euclidean Distance. Note that $D_0$ is classical Euclidean distance and $D_{-1}(x,y) = x + y$. Furthermore, $D_\rho$ is monotonic in the sense that for $\rho_1 \geq \rho_2$, it holds that $D_{\rho_1}(x,y) \leq D_{\rho_2}(x,y)$. }
The choice of our $CW(\rho)$ prediction intervals is motivated by the following theorem. 
\end{definition}
\begin{theorem}[Motivation and optimality under a perfect (asymptotic) scenario]\label{motivation_theorem}
Let $x\in\mathcal{X}$ and $\rho = \operatorname{cor}\big(Y(0), Y(1) \mid X = x\big) \in[-1, 1]$. Assume a perfect scenario:  $\big(Y(1), Y(0)\big)\mid X=x$ is Gaussian,  $\hat{\mu}_t(x) =  \mu_t(x)$ and suppose that we found conditionally valid prediction intervals:
\begin{align*}
 &\mathbb{P}\big(Y(t) \leq \hat{\mu}_t(x) + u_t(x)\mid X=x\big)  = 0.95, \quad  \mathbb{P}\big(Y(t) \geq \hat{\mu}_t(x) - l_t(x)  \mid X=x\big)=0.95, \,\,\,\, t=0,1.
\end{align*}
Then,  $CW(\rho)$ prediction intervals \eqref{D_rho_intervals_definition} are optimal in a sense that it is the smallest set satisfying:
\begin{align}\tag{\ref{def_cond_coverage}}
 \mathbb{P}\big(Y(1) - Y(0)\in C_{ITE}(X)\mid X=x\big)= 0.9.
\end{align}
\end{theorem}
The result in Theorem~\ref{motivation_theorem} holds under an idealized asymptotic scenario. In practice, estimation error or non-Gaussianity can lead to suboptimal performance, while additional assumptions can lead us to a different optimal prediction intervals. Nonetheless, the theorem provides valuable motivation: it shows that under ideal conditions, the $CW(\rho)$  construction yields the smallest valid prediction set for the individual treatment effect. In the next section, we provide theoretical guarantees under much weaker assumptions in finite samples.

\section{Theoretical guarantees}
\label{section5}
We show that 
$CW(\rho)$ intervals achieve valid finite-sample coverage for ITEs, provided a lower bound on $\rho$ and under mild, testable assumptions. The proofs are technically involved, drawing on tools from copula theory and sharp probabilistic inequalities to handle the challenges of reasoning under cross-world uncertainty. Formal proofs are provided in Appendix~\ref{appendix_proofs}.

In order to state any theoretical guarantees, we make a few simplifications:  1) we only consider upper prediction interval for the ease of notation. Generalizing this for both-sided intervals is straightforward. 2) Correlation $\rho(x) = \rho\in [-1,1]$ exists (we can generally work with a copula instead, but this would add a few technical and notational challenges). 3) We always use  $\hat{\tau}(x)= \hat\mu_1(x) - \hat\mu_0(x)$. 4) For simplicity of notation, we consider $1-\alpha = 0.9$. 

 In what follows, we distinguish between 'ground truth' $\rho$ and 'used hyperparameter' $\tilde \rho$ which can potentially differ. Note that for any $\tilde\rho \leq \rho$, $CW(\tilde\rho)$ bounds are wider than $CW(\rho)$ bounds. 
\begin{theorem}[Gaussian case]
\label{theorem_general_gaussian}
Let $\rho=\operatorname{cor}\big(Y(0), Y(1) \mid X = x\big)\in[-1,1]$. Suppose that we found marginally valid intervals:
\begin{equation}
\label{eq_in_theorem_coverage}
\mathbb{P}\big(Y(0) \geq \hat{\mu}_0(X) - l_0(X)\big) \geq 0.9, \quad \mathbb{P}\big(Y(1) \leq \hat{\mu}_1(X) + u_1(X)\big) \geq 0.9.
\end{equation}
Then, for any $\tilde\rho\leq \rho$, the $CW(\tilde\rho)$-bounds are marginally valid:
\begin{equation}
\label{eq_in_theorem_validity}
\mathbb{P}\bigg(Y(1)- Y(0) \leq  \hat{\tau}(X)  + D_{\tilde\rho}\big(l_0(X), u_1(X)\big)\bigg) \geq 0.9,
\end{equation}
provided that the following conditions hold:  
\begin{itemize}
    \item \textbf{Gaussianity}: $\big(Y(0), Y(1)\big)\mid X=x$ is normally distributed for all $x\in\mathcal{X}$, 
    \item \textbf{Coverage-variance relation}:   $\mathbb{E}[w(X)\phi_0(X)+\big (1-w(X)\big) \phi_1(X)] \geq 0.9,$ where  
   \[
\begin{aligned}
\phi_0(x) &:= \mathbb{P}\big(Y(0) \geq \hat{\mu}_0(X) -l_0(X) \mid X=x\big),\phi_1(x) := \mathbb{P}\big(Y(1) \leq \hat{\mu}_1(X) + u_1(X) \mid X=x\big), \\
w(x) &:= (1-\rho^2) \frac{\operatorname{Var}\big(Y(0)\mid X=x\big)}{\operatorname{Var}\big(Y(1)-Y(0)\mid X=x\big)} \in [0,1].
\end{aligned}
\]
    This condition holds, for example, under \textit{Homoskedasticity}, if  
    $\frac{\operatorname{Var}\big(Y(0)\mid X=x\big)}{\operatorname{Var}\big(Y(1)\mid X=x\big)} = \text{constant},$ if $\rho = \pm 1$, if $\phi_t(x)\geq 0.9$ or if $\phi_0(x)=\phi_1(x)$ for all $x\in\mathcal{X}$. 
    \item \textbf{$\boldsymbol{\hat\tau(x)}$ is not too biased} (this is relaxed in Section~\ref{section_CI}): $$\tau(x) - \hat\tau(x) = e_0(x) +e_1(x) \leq  D_{\tilde\rho}\big(l_0(x), u_1(x)\big) - D_{\rho}\big(l_0(x)-e_0(x), u_1(x)-e_1(x)\big),$$ where $e_1(x) =   \mu_1(x) - \hat  \mu_1(x) , e_0(x) =   -\mu_0(x) + \hat  \mu_0(x) $. This holds when $\rho = -1$.
\end{itemize}
\end{theorem}
Theorem~\ref{theorem_general_gaussian} demonstrates that $CW(\rho)$ intervals remain marginally valid in the Gaussian case, provided two assumptions hold. The coverage-variance relation protects us against the atypical scenario in which the error in conditional coverage is large in regions with high (relative) variance of the potential outcomes, and small in regions with low (relative) variance. Although this assumption is less interpretable, it is typically satisfied in common settings such as when $\frac{\operatorname{Var}\big(Y(0)\mid X=x\big)}{\operatorname{Var}\big(Y(1)\mid X=x\big)}$ is constant. Regarding the last assumption, we require control over the bias of the CATE estimator $\hat{\tau}$. We relax this assumption and provide more details in Section~\ref{section_CI}. 

The assumption of Gaussianity is not essential; in the following, we present results for important special cases $\rho = \pm 1$ without relying on any distributional assumptions. For brevity, we only present here the main message; full statement can be found in Appendix~\ref{section_appendix_theoretical_guarantees}.

\begin{theoremcustommmm}[Insights when $\rho=\pm 1$; Full Theorems in Appendix~\ref{section_appendix_theoretical_guarantees}]
\label{theorem_summary}
Suppose we have marginally valid prediction bounds \eqref{eq_in_theorem_coverage}. Then, the $CW(\rho)$ prediction intervals satisfy marginal coverage  \eqref{eq_in_theorem_validity} under either of the following assumptions:
\begin{itemize}
    \item \textbf{Shortest-width scenario:} \( \rho = 1 \) and \( \hat{\tau}(X) \) is not too biased.
    \item \textbf{Largest-width scenario:} \( \rho = -1 \) and at least one additional condition holds, such as:
    \begin{itemize}
        \item conditional coverage or $95\%$ marginal coverage of bounds in \eqref{eq_in_theorem_coverage},
        \item \( Y(t) \mid X \) has convex quantile tails, as is the case for most standard distributions including Gaussian, Uniform, and Gamma.
    \end{itemize}
    Additionally, the $CW(\rho)$ prediction intervals are tight without additional assumptions. 
\end{itemize}
\end{theoremcustommmm}
Theorems~\ref{Theorem_sum} and \ref{rho_1_theorem} show that strong finite-sample coverage guarantees can be achieved by imposing mild assumptions on the quality of either estimation of \( \hat{\tau}(X) \) or of the coverage in \eqref{eq_in_theorem_coverage}. All conditions in Theorems~\ref{Theorem_sum} and \ref{rho_1_theorem} are testable (except for the cross-world assumption on $\rho$).

\subsection{ When $\hat\tau(x)$ is strongly biased: introducing \texorpdfstring{$CW^{+CI}(\rho)$}{D+CI} prediction intervals}
\label{section_CI}

Theorems~\ref{theorem_general_gaussian} and \ref{theorem_summary} assume some control over the bias of the CATE estimator $\hat{\tau}$. This is in contrast with typical model-free guarantees from conformal inference. Unfortunately, some assumptions on $\hat{\tau}$ are necessary especially when $\rho=1$. If $Y_i(1) = \tau(X_i) + Y_i(0)$, optimal prediction intervals are simply confidence intervals for $\tau$, since $\mathbb{P}\big(Y(1)- Y(0) \in C_{ITE}(X)\big)  =  \mathbb{P}\big(\tau(X)\in C_{ITE}(X)\big)$; guarantees that conformal inference typically can not provide. 

To relax the assumptions on the bias of the CATE estimator, we enlarge $CW(\rho)$ prediction intervals by incorporating confidence intervals for $\hat{\tau}$ estimated, for instance, via bootstrapping.

\begin{definition}
Let $\rho\in [-1, 1]$. Consider the prediction intervals for $Y(1), Y(0)$ in the form~\eqref{C_1_and_C_0_definitions}. Denote an estimate of CATE as $ \hat{\tau}(x)$ with confidence intervals  $CI(x) = [\hat{\tau}(x)-r_l(x),\hat{\tau}(x)+r_u(x)]$. We define $CW^{+CI}(\rho)$ intervals as follows: 
\begin{equation*}
 \begin{aligned} 
CW^{+CI}(\rho):\,   C_{ITE}(x) =
   \Big[ \hat{\tau}(x)-cr_l(x) - D_\rho\left(l_1(x), u_0(x)\right), \,\,\, 
    \hat{\tau}_u(x) + cr_u(x)  + D_\rho\left(l_0(x), u_1(x)\right) \Big],
\end{aligned}
\end{equation*}
where $c\in[0,1]$ is a hyperparameter of our choice.  
\end{definition}
In simple terms, we extend  $CW(\rho)$ prediction intervals by adding c-times confidence intervals for $c\in[0,1]$. If $c=0$ then $CW^{+CI}(\rho) = CW(\rho)$. 

The bias assumption on $\hat{\tau}$ from Theorems~\ref{theorem_general_gaussian} and \ref{rho_1_theorem} is more strict when $\rho$ is large, while no assumption is required when $\rho = -1$. This motivates the following choice of hyperparameter: $c = \frac{1+\rho}{2}$. For $\rho = -1$ we have  $CW^{+CI}(\rho) = CW(\rho)$, while for $\rho = 1$ we add full confidence intervals to the original $CW(\rho)$. Empirically, we found that a choice of hyperparameter $c = (\frac{1+\rho}{2})^2$ works better in practice. 

\begin{theorem}[Conditional validity of \( CW^{+CI}(\rho)\) intervals when \( \rho = 1 \)]
\label{theorem_CI}
Let \( x \in \mathcal{X} \) and assume $\rho(x) = \operatorname{cor}\big(Y(0), Y(1)\mid X = x\big) = 1$.  Suppose that for some \( \beta \in (0,1) \), our confidence intervals satisfy \( \mathbb{P}\big(\tau(x) \leq \hat{\tau}(x) + r(x)\big) \geq 1 - \beta \). For any $\tilde\rho\in[-1,1]$ holds:

\begin{itemize}
    \item If \( \operatorname{var}(Y(0)\mid X = x) = \operatorname{var}(Y(1) \mid X = x) \),   then \( CW^{+CI}(\rho) \) prediction intervals are conditionally valid:
\[
\mathbb{P}\big(Y(1)-Y(0) \leq \hat{\tau}(X) + r(X) + D_{\tilde\rho}\big(l_0(X), u_1(X)\big) \mid X = x\big) \geq 1 - \beta.
\]
\item  If \( \operatorname{var}\big(Y(0) \mid X = x\big) \neq \operatorname{var}\big(Y(1) \mid X = x\big) \), and the following conditions hold:
    \begin{itemize}
        \item $    \mathbb{P}\big(Y(0) \geq \hat{\mu}_0(X) -l_0(X) \mid X = x\big) = 1 - \alpha=\mathbb{P}\big(Y(1) \leq \hat{\mu}_1(X) + u_1(X) \mid X = x\big),$
        \item $    \mathbb{P}\left(\mu_0(x) \geq \hat{\mu}_0(x) - \tfrac{1}{2} r(x)\right) \geq 1 - \tfrac{1}{2} \beta$, and $\mathbb{P}\left(\mu_1(x) \leq \hat{\mu}_1(x) + \tfrac{1}{2} r(x)\right) \geq 1 - \tfrac{1}{2} \beta$,
    \end{itemize}
    then \( CW^{+CI}(\rho) \) prediction intervals are conditionally valid:
 \[
    \mathbb{P}\bigg(Y(1) -Y(0)\leq \hat{\tau}(X) + r(X) + D_{\tilde{\rho}}\big(l_0(X), u_1(X)\big) \mid X = x\bigg) \geq (1 - \alpha)(1 - \beta).
    \]
\end{itemize}
\end{theorem}

While Theorem~\ref{theorem_CI} focuses on conditional coverage, a similar result can be also stated for marginal coverage by integrating over $X$. We have shown that \( CW(\rho) \) and \( CW^{+CI}(\rho) \) intervals enjoy strong theoretical guarantees; in the following section, we also demonstrate their strong empirical performance.

\section{Numerical experiments }
\label{section_simulations}

In this section, we present a series of numerical experiments to evaluate the empirical performance of our proposed prediction intervals. Since the evaluation is possible only when both potential outcomes are known, we focus on synthetic datasets and semi-synthetic (real-data-inspired) benchmarks with known ground truth. We examine a variety of scenarios by varying the dimensionality, noise distribution, and the cross-world correlation parameter $\rho$, and benchmark our approach against several state-of-the-art methods.

A user-friendly implementation of our methods in both \texttt{R} and \texttt{Python}, along with scripts to reproduce all experiments, is available at: \url{https://github.com/jurobodik/ITE_prediction_cross_world.git}.

\subsection{Details of the experiments}
\label{Section5.1}

\textbf{Synthetic datasets: }We consider a variety of data-generating processes similar as in the related literature; full details are provided in Appendix~\ref{appendix_simulations_data}. The \textbf{synthetic} datasets feature non-constant propensity scores and randomly generated CATE functions based on random smooth polynomials. These settings allow us to vary the dimensionality \(d = \dim(\mathbf{X})\), the cross-world correlation parameter \(\rho\), the distribution of the noise, and the signal-to-noise ratio. In addition, we include the \textbf{IHDP} dataset, which uses real covariates and simulated counterfactual outcomes. Figure~\ref{fig_rho_d_1} illustrates the synthetic setup for \(d = 1\), along with our proposed \(CW(\rho)\) prediction intervals.

\textbf{Implementation of our methods:} To implement the proposed  $CW(\rho)$ and $CW^{+CI}(\rho)$ intervals, we use the classical version of CQR to construct prediction intervals \eqref{C_1_and_C_0_definitions}. While several alternatives, such as CLEAR \citep{azizi2025clearcalibratedlearningepistemic} or PCS-UQ \citep{agarwal2025pcs, Yu_2024_Vertidical_Data_Science_Book} often achieve better performance, we opt for the simplicity of CQR. This choice also aligns with prior work, such as \cite{Lei_2021_Conformal_Inference}, which adopts a similar baseline. For quantile regression in low-dimensional case ($d\leq 5$), we apply qGAM \citep{qGAM}, while in higher dimensions ($d > 5$), we use quantile random forests \citep{Meinshausen_2006_Quantile_Regression}. This choice is mainly motivated by the computational complexity;  qGAM has often superior accuracy in low-dimensional settings, however random forests has much better computational complexity. Alternative is to use TabPFN \citep{hollmann2023tabpfntransformersolvessmall}. 

\textbf{Baseline methods: } In Appendix~\ref{lit_review}, we provide literature review of existing methods that estimate prediction intervals for ITEs. In our experiments, we consider two methods from \cite{Lei_2021_Conformal_Inference}: the Exact Naïve and the Inexact Nested approaches. We do not include the Exact Nested method from the same source, as its performance is comparable to or worse than the Naïve method. Additionally, we consider the Conformal Monte Carlo (CMC) method introduced by \cite{jonkers2024conformalconvolutionmontecarlo} and doubly-robust (DR) conformal meta-learner from \cite{Alaa_2023_Conformal_Meta_learners}. In all cases, we use random forests as the base estimators.
 \begin{figure}[t]
    \centering
    \includegraphics[width=1\linewidth]{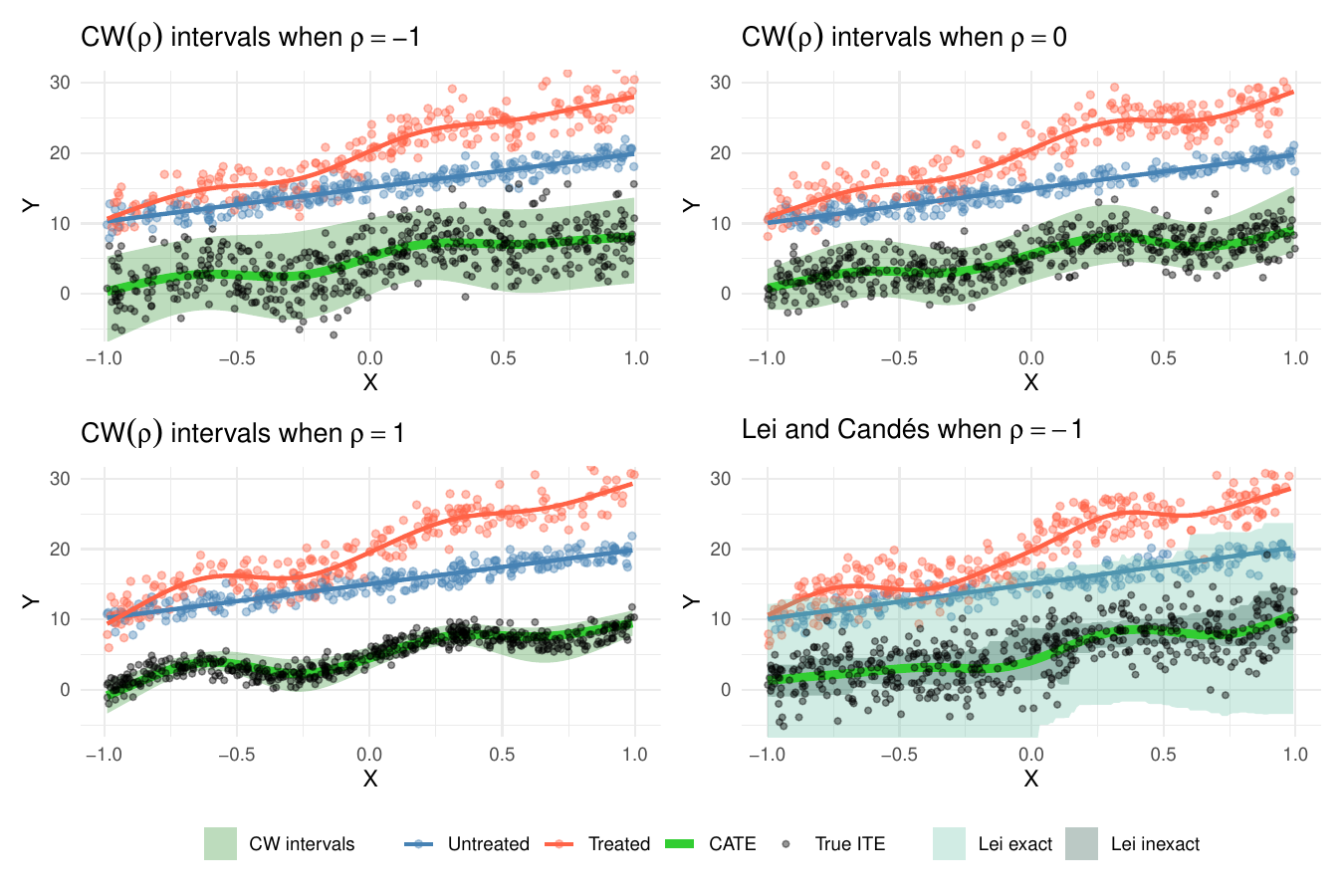}
    \caption{Synthetic $d=1$ dataset described in Section~\ref{Section5.1}. The variance of ITE drastically decreases with larger $\rho$, and our $CW(\rho)$ intervals seem to have correct coverage. We also include exact and inexact prediction method from \cite{Lei_2021_Conformal_Inference} as a comparison. Note that the observed values (treated and untreated units) are the same on all four figures.   }
    \label{fig_rho_d_1}
\end{figure}

\subsection{Evaluation of the experiments}
Figure~\ref{fig_sim1} shows results on synthetic data. Results on the IHDP dataset are similar to the 15-dimensional synthetic case and can be found in Appendix~\ref{appendix_simulations_additional_experiments}. Additional experiments with non-Gaussian noise are also located in this appendix. 

\textbf{Here are the key takeaways:}
\begin{itemize}
    \item All competing methods exhibit significant under-coverage when $\rho \leq -0.5$ and over-coverage when $\rho \geq 0.5$ (with the exception of the Naive method of course). In contrast, our proposed $CW(\rho)$ and $CW^{+CI}(\rho)$ intervals maintain valid coverage in all scenarios when the true value (lower bound) of $\rho$ is known. 
    
\item $CW^{+CI}(\rho)$ achieves more than 3× better performance in terms of Coverage--Width loss (a combination of interval width and coverage error; see Appendix~\ref{appendix_simulations_additional_experiments}) when $\rho = \pm 1$, while still outperforming other methods even when $\rho$ is slightly misspecified or equal to zero. When $\rho = 1$, our $CW^{+CI}(\rho)$ intervals achieve correct coverage with intervals more than 2-3× shorter than competitors. Among the other methods, CMC~\citep{jonkers2024conformalconvolutionmontecarlo} performs best when $\rho \leq 0$, and DR~\citep{Alaa_2023_Conformal_Meta_learners} performs best when $\rho > 0$. This is summarized in Figure~\ref{figure_result_coverage_width_loss} in Appendix~\ref{appendix_simulations_additional_experiments}.

    \item In low-dimensional settings, where the treatment effect $\tau$ is well-estimated, the $CW(\rho)$ intervals achieve a perfect coverage. However, in higher-dimensional settings ($d = 15$) and when $\rho \geq 0.5$, the $CW(\rho)$ intervals tend to under-cover. This aligns with our theory from Section~\ref{section5}, as higher dimensionality (with fixed sample size) leads to less accurate estimation of $\tau$. In all cases, even if $\rho \geq 0.5$, the $CW^{+CI}(\rho)$ prediction intervals achieve almost perfect coverage across all situations, with only a slight increase in width (approximately 1-8\% wider) compared to the $CW(\rho)$ intervals. \textbf{In conclusion, we recommend using the $\boldsymbol{CW^{+CI}(\rho)}$ prediction intervals}. However, the main drawback of the $CW^{+CI}(\rho)$ intervals is their higher computational cost due to bootstrapping, whereas the $CW(\rho)$ intervals are exceptionally fast.

       \item Our methodology appears to be quite robust to deviations from Gaussianity: in the copula-based simulations presented in Appendix~\ref{appendix_simulations_additional_experiments}, the \( CW(\rho) \) intervals consistently achieve near-nominal coverage across a  range of copula and marginal distribution combinations.   
    \item In real-world scenarios, the true value of \(\rho\) is typically only lower-bounded. In our simulations, we constructed prediction intervals using a misspecified \(\rho_{\text{used}} = \rho_{\text{true}} - 0.25\). As expected, this led to slightly increased coverage and width (by approximately 10–20\%). Nevertheless, these intervals still strongly outperformed methods that do not incorporate \(\rho\), highlighting the benefit of leveraging even partial information about cross-world dependence.

\end{itemize}

Overall, knowledge of (lower bound of) $\rho$ allowed us to construct prediction intervals that are both valid and efficient, adapting to the underlying structure of the data. Even when $\rho$ is not known exactly, our method enables the use of $\rho$ as a sensitivity parameter—allowing practitioners to explore hypothetical scenarios and assess robustness to unobserved covariates. This type of analysis is not possible with existing methods.

\begin{figure}
    \centering
    \hspace*{-0.06\linewidth}\includegraphics[width=1.07\linewidth]{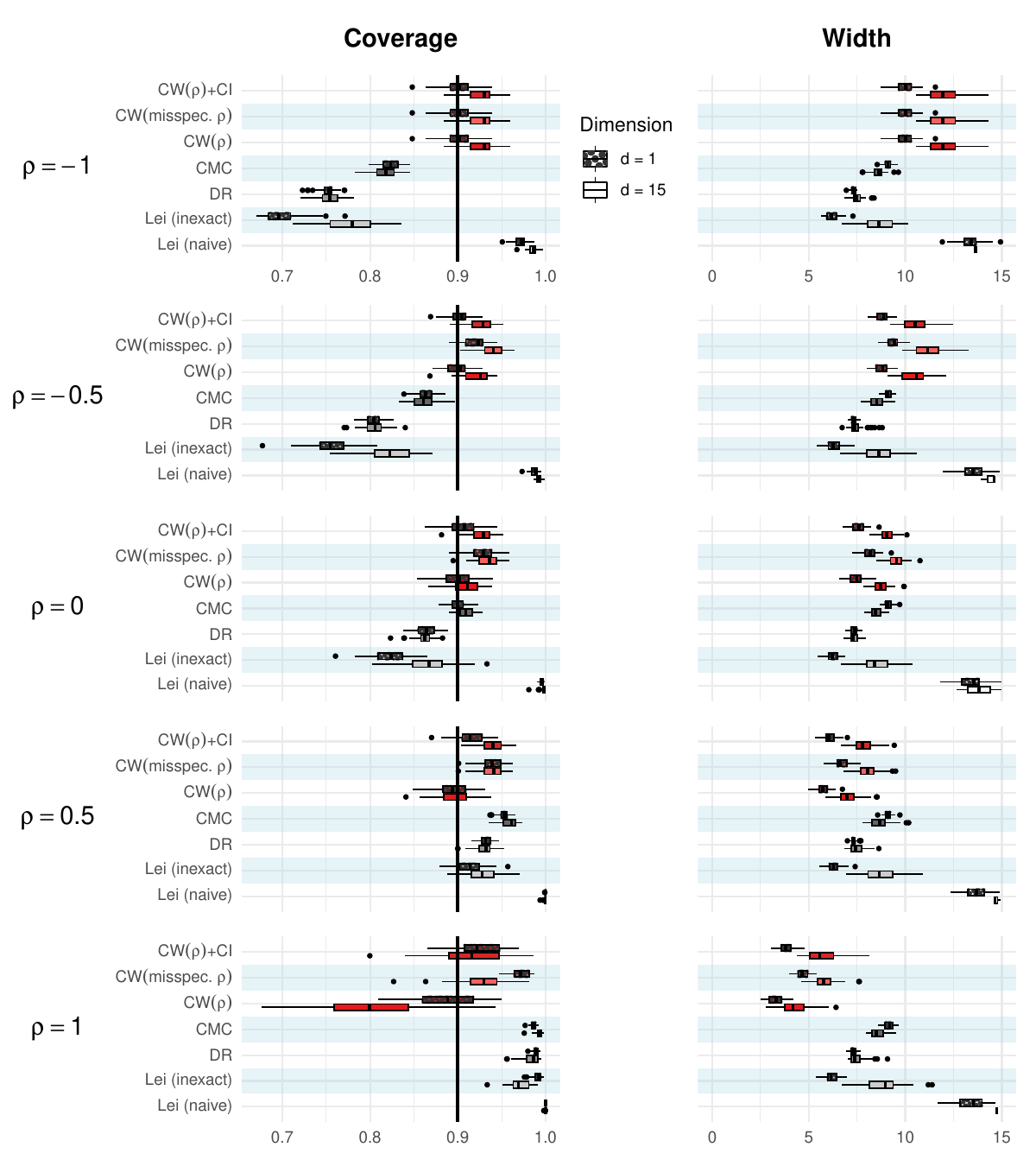}
    \caption{Box plots comparing the coverage and average interval width of different estimation methods across cross-world correlation levels $\rho \in \{-1, -0.5, 0, 0.5, 1\}$, based on 50 synthetic datasets with $n = 2000$ and $d=1$ and $d=15$. Methods include Lei’s exact and inexact intervals \citep{Lei_2021_Conformal_Inference}, the conformal DR-meta-learner \citep{Alaa_2023_Conformal_Meta_learners}, CMC \citep{jonkers2024conformalconvolutionmontecarlo}, and three versions of our $CW(\rho)$ intervals: the baseline $CW(\rho)$, a misspecified variant $CW(\text{misspec.}\rho)$ using $\tilde\rho = \rho - 0.25$ (capped at $-1$), and one with added conformal confidence intervals ($CW^{+CI}(\rho)$). The vertical line marks the desired $90\%$ coverage.}
    \label{fig_sim1}
\end{figure}

\section{Conclusion and Future Research}

Uncertainty quantification for ITEs is critical in high-stakes applications, especially when relying on black-box machine learning models. Despite its importance, consistent estimation of ITEs has remained fundamentally elusive due to their inherent unidentifiability. In this work, we introduce a novel cross-world assumption that enables consistent estimation of ITEs. Leveraging expert-informed reasoning about a lower bound on the cross-world parameter $\rho$, we propose a new class of prediction intervals, denoted $CW(\rho)$, and establish their asymptotic optimality alongside finite-sample guarantees. Our method substantially outperforms state-of-the-art baselines both theoretically and empirically, particularly when $\rho \approx \pm 1$. While $\rho$ is itself unidentifiable, we present several scenarios and domain-specific arguments where a conservative lower bound can be credibly assumed.

A promising direction for future research is to integrate the cross-world parameter $\rho$ into meta-learning frameworks and pseudo-outcome estimation strategies. Another compelling extension is to generalize $\rho$ to settings with continuous treatments \citep{schroder2024conformalpredictioncausaleffects, bodik_extrapolating_extreme}, to causal inference in time series \citep{NEURIPS2023_47f2fad8, bodikpasche2024grangercausalityextremes}, or combining it with Fréchet–Hoeffding-type copula bounds \citep{zhang2025individualtreatmenteffectprediction}.
Finally, although we often assumed a Gaussian copula structure between potential outcomes, alternative copulas, such as extreme value copulas, may be more suitable in applications focusing on tail behavior or rare events \citep{Deuber}. 

A key remaining challenge is the lack of empirical datasets in which the cross-world parameter $\rho$ can be directly evaluated or credibly inferred. We encourage further efforts to collect or annotate data that support expert-informed reasoning about $\rho$.

\section*{Conflict of interest and data availability}
The open-source implementation of the methods discussed in this manuscript together with the data used can be found in the supplementary package or at \url{https://github.com/jurobodik/ITE_prediction_cross_world.git}.

The authors declare that they have no known competing financial interests or personal relationships that could have appeared to influence the work reported in this paper.

\section*{Acknowledgments}
JB was supported by the Hasler-Stiftung foundation under grant number 24009. YH acknowledges support from the Citadel Securities PhD Fellowship. The authors would like to thank Valérie Chavez-Demoulin for continuous support.

\clearpage 
\bibliography{bibliography}
\newpage
\appendix
\pagenumbering{arabic}
\renewcommand*{\thepage}{A\arabic{page}}
\begin{center}
\Large \textbf{Appendix}
\end{center}
This appendix is organized as follows: \begin{itemize}
\item Appendix~\ref{appendix_rho_estimation} discusses the estimation procedure for $\rho_L(x)$ introduced in Section~\ref{section_2}, omitted from the main text for clarity. 
\item Appendix~\ref{section_appendix_theoretical_guarantees} extends the theoretical results presented in Section~\ref{section5}, omitted from the main text for clarity. 
\item Appendix~\ref{lit_review} provides a literature review of existing methods for estimating prediction intervals for individual treatment effects. \item Appendix~\ref{appendix_simulations_data} describes in detail the datasets used in our experiments in Section~\ref{section_simulations}. 
\item Appendix~\ref{appendix_simulations_additional_experiments} presents additional experiments, including results on the semi-synthetic IHDP dataset and a setting with non-Gaussian noise.
\item Appendix~\ref{appendix_proofs} contains all theoretical results and their corresponding proofs. \end{itemize}

\section{ Data-driven estimation of $\rho$ when auxiliary covariates are available}
\label{appendix_rho_estimation}

In Section~\ref{section_rho_with_Z}, we showed that $\rho(x)$ is lower bounded by either
$\tilde\rho_L(x)$ or $\rho_L(x)$, where
\[
\tilde\rho_L(x)=\frac{A(x)-B_L(x)}{\sigma_0(x)\sigma_1(x)},\qquad
\rho_L(x)=\frac{A(x)}{\sigma_0(x)\sigma_1(x)}.
\]
Here, we discuss estimation of this identifiable lower bound.

\subsection{Direct estimation of $\rho_L(x)$}
First, estimate the conditional means: \(
\widehat{\mu}_0(x,z),
\widehat{\mu}_1(x,z).
\) Here, any method can be used from classical parametric regression to flexible machine-learning estimators (e.g., random forests, gradient boosting, neural networks). Then estimate
\begin{equation}\label{first_equation_in_appendix_A}
\widehat{A}(x)=\widehat{\mathrm{Cov}}\{\widehat{\mu}_0(x,Z),\widehat{\mu}_1(x,Z)\mid X=x\},
\end{equation}
where the conditioning on $X=x$ may be implemented by binning or kernel smoothing. The marginal conditional variances $\sigma_t^2(x)$ can be estimated from each arm as
$\widehat{\sigma}_t^2(x)=\widehat{\mathrm{Var}}(Y\mid T=t,X=x)$ using standard regression or smoothing methods (ideally using the same method as for the estimation of the conditional means). Potentially, if we are not willing to assume non-negativity of $B(x)$, we estimate $\widehat{B}_L(x)=\widehat{\mathbb{E}}\!\left[\sqrt{\widehat{\sigma}_0(x,Z)\widehat{\sigma}_1(x,Z)}\,\middle|\, X=x\right]$ again by using some regression or flexible machine-learning estimator. Plugging these quantities into \eqref{eq:rho_lower} yields the data-driven lower bound $\widehat{\rho}_L(x)$.
\begin{lemma}[Idea; Full statement in Lemma~\ref{lem_appendix_rho_lower_consistency}]
If the nuisance parameters are consistently estimated, then $\widehat{\rho}_L(x)\overset{P}{\to} \rho_L(x)$ as $n\to\infty$. Consequently, for any $\epsilon>0$,
\begin{align}
\lim_{n\to\infty}\Pr\!\left(\rho(x)\ \ge\ \widehat{\rho}_L(x)-\epsilon\right)=1.
\label{eq:rho_lower_asymp_ineq}
\end{align}
\end{lemma}

\subsection{Imputation-based estimator.}
We now describe an alternative approach that estimates cross-world dependence by imputing counterfactual outcomes for the subset of units for which $Z$ is observed. The key idea is to \emph{impute the missing potential outcome} using $(X,Z)$, thereby constructing a pseudo-pair that approximates $(Y(0),Y(1))$ at the unit level.

There are many ways to impute $\widehat Y_i(1-T_i)$; see \cite{bodik2026retrospective} for details. Common options include matching \citep{abadie2006matching_imputation} and adversarial generative modeling \citep{yoon2018ganite, li2026controllable}. A simple approach is outcome-regression imputation:
\begin{align}
\widehat Y_i(1-T_i)
:=
\widehat\mu_{1-T_i}(X_i,Z_i),
\label{eq:Yhat_counterfactual_refined}
\end{align}
where $\widehat\mu_t(x,z)$ estimates $\mu_t(x,z)=\E[Y(t)\mid X=x,Z=z]$ via random forest, splines or neural networks. Another option is to impute via a CATE model:
\begin{align}
\widehat Y_i(1-T_i)
:=
Y_i + (1-2T_i)\widehat\tau(X_i,Z_i),
\label{eq:ert_refined}
\end{align}
where $\widehat\tau(x,z)$ estimates $\tau(x,z)=\E[Y(1)-Y(0)\mid X=x,Z=z]$. (Equivalently, this sets $\widehat Y_i(1)=Y_i+\widehat\tau(X_i,Z_i)$ when $T_i=0$ and $\widehat Y_i(0)=Y_i-\widehat\tau(X_i,Z_i)$ when $T_i=1$.)

Given the observed outcome and its imputed counterpart, we form the pseudo-potential-outcome pair
\begin{align}
\widetilde Y_i(1) := T_i Y_i + (1-T_i)\widehat Y_i(1),
\qquad
\widetilde Y_i(0) := (1-T_i) Y_i + T_i\widehat Y_i(0),
\label{eq:pseudo_pair_refined}
\end{align}
so that $\widetilde Y_i(T_i)=Y_i$ is observed and $\widetilde Y_i(1-T_i)$ is imputed.

We then estimate the cross-world parameter by
\begin{align}
\hat{\rho}_{\mathrm{imp}}(x)
:=
\mathrm{Cor}\{\widetilde Y(0),\widetilde Y(1)\mid X=x\},
\label{eq:cov_tilde_refined}
\end{align}
where conditioning on $X=x$ may be implemented via binning, kernel smoothing, or a regression of $\widetilde Y(0),\widetilde Y(1)$ on $X$ combined with estimates of the conditional variances.

\paragraph{Discussion.}
In contrast to the bound-based approach in \eqref{eq:rho_lower}, $\widehat\rho_{\mathrm{imp}}(x)$ is generally \emph{not} guaranteed to be a lower bound on $\rho(x)$ without further assumptions on the imputation error. This estimator effectively replaces the missing potential outcome with a proxy, so its accuracy depends on how well $(X,Z)$ predicts both $Y(0)$ and $Y(1)$ and on the stability of the imputation model across treatment arms. Nevertheless, when $(X,Z)$ is highly predictive, the pseudo-pairs in \eqref{eq:pseudo_pair_refined} can yield a stable, data-driven estimate of cross-world association that is useful for choosing a principled default value of $\rho(x)$ in the CW construction and for guiding sensitivity analyses.
In practice, we recommend using \emph{multiple} imputation strategies and aggregating the resulting estimates, e.g.,
\[
\widehat\rho_{\mathrm{imp}}^{,\mathrm{ens}}(x)=
\frac{1}{M}\sum_{m=1}^M \widehat\rho_{\mathrm{imp}}^{(m)}(x),
\]
where $\widehat\rho_{\mathrm{imp}}^{(m)}(x)$ is computed from pseudo-pairs formed using the $m$th imputation method. Different imputation procedures encode different modeling assumptions and therefore exhibit different biases and implicit ways of handling the missing counterfactual. Ensembling across several reasonable imputers can reduce method-specific artifacts and yield a more stable, data-driven default for $\rho(x)$ in the CW construction, while still allowing sensitivity analysis around the aggregated choice.

\subsection{Simulations about the estimation of $\rho(x)$}
\label{appendix_simulations_rho}

\subsubsection{Synthetic data from Section~\ref{subsubsection2.21}}

We study the tightness of the proposed lower bounds for the cross-world parameter \(\rho(x)\) under a simple randomized design. We generate $n=1000$ potential outcomes
\[
Y(0)=\beta_0X+\gamma Z_1+\varepsilon_0,\qquad
Y(1)=\beta_1X+\gamma Z_1+\varepsilon_1,\qquad \mathrm{Cor}(\varepsilon_0,\varepsilon_1)=\rho_\varepsilon,
\]
with \(X,Z_1, \varepsilon_0,\varepsilon_1\stackrel{\text{iid}}{\sim}\mathcal{N}(0,1)\) and $\beta_0, \beta_1\stackrel{\text{iid}}{\sim}{UNIF}(0.5, 1.5)$. Treatment is assigned as \(T\sim\mathrm{Bernoulli}(1/2)\) and the observed outcome is \(Y=TY(1)+(1-T)Y(0)\). 

We sweep \(\gamma\) to control the strength of shared structure across arms, and consider \(\rho_\varepsilon\in\{-0.5,0,0.5\}\). For each \((\gamma,\rho_\varepsilon)\), we estimate \(\hat\rho_L(x)\) and \(\hat{\tilde\rho}_L(x)\) using linear regression nuisance fits and kernel smoothing over \(X\). “Truth” \(\rho(x)\) is computed directly from simulated \((Y(0),Y(1))\). 

We summarize tightness by the gap \(\rho(x)-\hat\rho(x)\) averaged over an evaluation grid in \(x\). Plot~\ref{fig_rho_vs_rho} shows the mean gap as a function of \(\gamma\), with shaded bands given by the empirical 5th and 95th percentiles of the gap across $R=100$ Monte Carlo replications. 

\textbf{Conclusion: }As the predictive power of $Z$ increases, the gap between the true $\rho(x)$ and its estimate shrinks. As expected, \(\widehat{\tilde\rho}_L(x)\) is highly conservative. In contrast, when $\rho_\varepsilon<0$, \(\widehat{\rho}(x)\) can exceed $\rho(x)$ and therefore fails to be a valid lower bound.

\begin{figure}
    \centering
\includegraphics[width=0.95\linewidth]{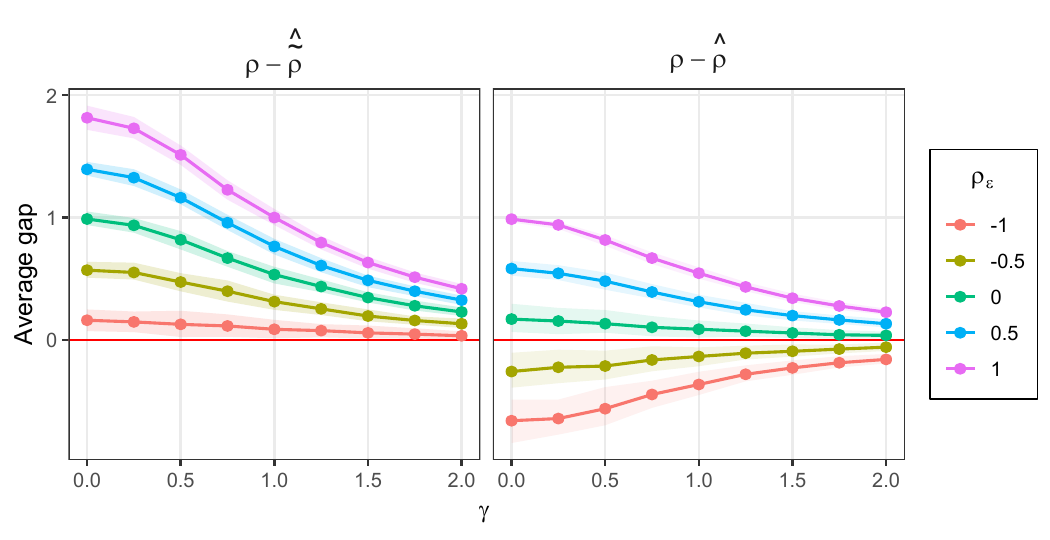}
 \caption{Average gap between the true cross-world correlation and its estimator as a function of the predictive power of the auxiliary variable $\gamma$. The left panel reports $\mathrm{GAP}=\mathbb{E}_X\!\left[\rho(X)-\widehat{\tilde\rho}(X)\right]$, while the right panel reports $\mathrm{GAP}=\mathbb{E}_X\!\left[\rho(X)-\widehat{\rho}(X)\right]$. Shaded bands are $95\%$ confidence intervals computed from repeated simulations as described in Appendix~A; the horizontal line at $0$ indicates no bias.}
    \label{fig_rho_vs_rho}
\end{figure}

\subsubsection{Twins dataset: details about $\rho$ estimation}
\label{appendix_twins_rho}

In the Twins benchmark we observe covariates $X=(X_1,\dots,X_{51})$  such as birth weight and gestational age, and both binary potential outcomes $(Y(0),Y(1))$  representing mortality within the first year of life.

For the bound estimation we split covariates into $X=(X_1,X_2)$, which we treat as the low-dimensional features at which $\rho(x)$ is evaluated, and auxiliary covariates $Z=(X_3,\dots,X_{51})$, used only to learn $\mu_t(x,z)=\mathbb{E}(Y\mid T=t,X=x,Z=z)$. We fit two arm-specific nuisance models for $\mu_t$ on the factual data in each treatment arm, using either probabilistic random forests (\texttt{ranger}, 500 trees) or a linear probability model (clamped to $[0,1]$). To reduce sparsity from conditioning on discrete covariates, we evaluate only at the three most frequent observed $(X_1,X_2)$ combinations.

For each evaluation point $x$, we compute kernel weights
$w_i(x)=\exp\!\left(-\frac{\|X_i-x\|^2}{2h_A^2}\right),$ where
$h_A=0.5$, and form predictions $\widehat\mu_0(x,Z_i)$ and $\widehat\mu_1(x,Z_i)$ by fixing $X=x$ while letting $Z$ vary across the sample. We estimate
$A(x)=\mathrm{Cov}\{\mu_0(x,Z),\mu_1(x,Z)\mid X=x\}$ by the weighted empirical covariance of these predicted probabilities. We estimate $\sigma_t(x)=\mathrm{Var}(Y\mid T=t,X=x)$ by a local weighted variance within each treatment arm using the same Gaussian kernel form (bandwidth $h_\Sigma=0.5$), and then compute
$\widehat\rho_L(x)=\frac{\widehat A(x)}{\sqrt{\widehat\sigma_0(x)\widehat\sigma_1(x)}},$
using small $\varepsilon$ stabilizers and clipping to $[-1,1]$ for numerical robustness. For the refined bound we additionally estimate
$B_L(x)=\mathbb{E}[\sqrt{\sigma_0(x,Z)\sigma_1(x,Z)}\mid X=x]$ using the binary-outcome identity
$\sigma_t(x,z)=\mu_t(x,z)\{1-\mu_t(x,z)\}$, i.e.,
\[
\widehat B_L(x)=
\frac{\sum_i w_i(x)\sqrt{\widehat\mu_0(x,Z_i)\{1-\widehat\mu_0(x,Z_i)\}\,\widehat\mu_1(x,Z_i)\{1-\widehat\mu_1(x,Z_i)\}}}{\sum_i w_i(x)},\,\,
\widehat{\tilde\rho}_L(x)=
\frac{\widehat A(x)-\widehat B_L(x)}{\sqrt{\widehat\sigma_0(x)\widehat\sigma_1(x)}}.
\]
Because Twins provides both $Y(0)$ and $Y(1)$, we also report an oracle estimate of $\rho(x)$ computed directly from $(Y(0),Y(1))$. Table~\ref{tab:twins-rho} shows a typical run, and similar qualitative behavior is observed under other choices of the pair $(X_i,X_j)$.

\begin{table}[h!]
\centering

\begin{tabular}{llrrrrr}
\hline
Pair & Strata & $n$ & $\widehat\rho_L(x_i, x_j)$ & $\widehat{\tilde\rho}_L(x_i, x_j)$ & $\widehat\rho_{\mathrm{oracle}}(x_i, x_j)$ \\
\hline
$(X_1,X_2)$ & 1 & 11630 & 0.417 & -0.108 & 0.647 \\
$(X_1,X_2)$ & 2 &   246 & 0.407 & -0.156 & 0.662 \\
$(X_1,X_2)$ & 3 &    40 & 0.463 & -0.021 & 0.661 \\
\hline
$(X_1,X_3)$ & 1 &  1233 & 0.407 & -0.130 & 0.698 \\
$(X_1,X_3)$ & 2 &   882 & 0.428 & -0.076 & 0.622 \\
$(X_1,X_3)$ & 3 &   789 & 0.353 & -0.249 & 0.599 \\
\hline
$(X_2,X_3)$ & 1 &  1215 & 0.402 & -0.143 & 0.694 \\
$(X_2,X_3)$ & 2 &   876 & 0.424 & -0.084 & 0.616 \\
$(X_2,X_3)$ & 3 &   774 & 0.349 & -0.250 & 0.596 \\
\hline
\end{tabular}
\caption{Estimates of $\rho$ on the most frequent pairs of covariates.}
\label{tab:twins-rho}
\end{table}

\subsubsection{Cross-over studies}
\label{appendix_Cross_over_studies}

\begin{table}[t]
\centering
\caption{Estimated empirical cross-world correlations across thirty cross-over datasets. ``Measured'' denotes datasets of non-simulated origin of public-domain datasets; ``Simulated'' denotes widely-used datasets that are simulated, heavily edited or derived from simulated datasets.}
\label{tab_cross_over_appendix}
\small
\begin{tabular}{clrrrr}
\toprule
Dataset number & Type & $n$ & $\rho$ & 95\% CI lower & 95\% CI upper \\
\midrule
09 & Simulated & 222 & 0.9920 & 0.9896 & 0.9939 \\
16 & Measured  & 38  & 0.9653 & 0.9338 & 0.9820 \\
21 & Measured  & 77  & 0.9072 & 0.8574 & 0.9401 \\
01 & Measured  & 77  & 0.8972 & 0.8426 & 0.9336 \\
06 & Measured  & 77  & 0.8972 & 0.8426 & 0.9336 \\
28 & Measured & 64  & 0.8853 & 0.8174 & 0.9290 \\
03 & Measured  & 76  & 0.8844 & 0.8231 & 0.9254 \\
30 & Simulated & 11  & 0.8842 & 0.6056 & 0.9697 \\
10 & Measured  & 18  & 0.8667 & 0.6716 & 0.9494 \\
25 & Simulated & 70  & 0.8611 & 0.7850 & 0.9116 \\
17 & Measured  & 19  & 0.8478 & 0.6400 & 0.9400 \\
08 & Simulated & 222 & 0.8378 & 0.7937 & 0.8732 \\
22 & Simulated & 42  & 0.8263 & 0.6976 & 0.9034 \\
24 & Measured  & 39  & 0.8261 & 0.6906 & 0.9056 \\
11 & Measured  & 37  & 0.8176 & 0.6715 & 0.9025 \\
13 & Simulated & 222 & 0.8135 & 0.7636 & 0.8537 \\
15 & Simulated & 222 & 0.8135 & 0.7636 & 0.8537 \\
27 & Simulated & 155 & 0.8123 & 0.7508 & 0.8598 \\
05 & Measured  & 26  & 0.8074 & 0.6112 & 0.9101 \\
02 & Measured  & 24  & 0.7900 & 0.5674 & 0.9050 \\
29 & Simulated & 12  & 0.7817 & 0.3768 & 0.9358 \\
23 & Measured  & 22  & 0.7193 & 0.4273 & 0.8754 \\
07 & Simulated & 360 & 0.7055 & 0.6495 & 0.7539 \\
04 & Measured  & 51  & 0.6343 & 0.4347 & 0.7745 \\
26 & Measured  & 54  & 0.5749 & 0.3630 & 0.7302 \\
18 & Simulated & 60  & 0.3600 & 0.1168 & 0.5625 \\
19 & Simulated & 60  & 0.3600 & 0.1168 & 0.5625 \\
20 & Simulated & 60  & 0.3437 & 0.0983 & 0.5496 \\
14 & Simulated & 76  & 0.3405 & 0.1246 & 0.5256 \\
12 & Simulated & 77  & 0.1748 & -0.0512 & 0.3837 \\
\midrule
Average & -- & -- & 0.7409 & -- & -- \\
\bottomrule
\end{tabular}
\end{table}

Table~\ref{tab_cross_over_appendix} reports empirical paired correlations computed from the 30 reference datasets distributed in the \texttt{replicateBE} package \citep{Schutz2020ReplicateBE, replicateBE_package}. These datasets arise from crossover studies, where subjects receive the reference formulation ($R$), the test formulation ($T$), or both across multiple periods. The collection was assembled primarily as a public benchmark for validating software used in replicate-design bioequivalence analyses.  It contains a mixture of public-domain datasets, edited variants, and simulated examples \citep{replicateBE_package}.

For example, the first dataset n.01 is a 4-period replicate bioequivalence study, in which each subject receives both the reference formulation $R$ and the test formulation $T$ across different periods. Here, the outcome is a pharmacokinetic response, denoted by \texttt{PK}. In simple words, this outcome measures how strongly the subject was exposed to the drug after receiving a given formulation, using a standard bioequivalence endpoint such as a peak concentration or overall exposure measure. Thus, $Y(0)$ denotes the outcome the subject would have under the reference treatment $R$, and $Y(1)$ denotes the outcome under the test treatment $T$.

For each dataset, we construct a subject-level pair $(Y(0), Y(1))$ by averaging repeated observations under treatment $R$ and treatment $T$ within subject whenever both are available, and then compute the direct paired empirical correlation $\hat\rho = \mathrm{cor}(Y(0), Y(1))$. Confidence intervals are obtained from the Fisher $z$-transform. Because the underlying study designs differ, including full replicate, partial replicate, and related crossover layouts, these estimates should be viewed as descriptive proxies rather than literal identifications of the cross-world correlation in our framework \citep{replicateBE_package}. Nevertheless, Table~\ref{tab_cross_over_appendix} shows that the estimated correlations are typically positive and large, which is consistent with the idea that stable latent subject characteristics may induce substantial positive dependence between outcomes under the two treatment states.

\subsubsection{High-dimensional simulations about the estimation of $\rho(x)$ }
\label{appendix_high_dim_est_rho}

We generate i.i.d.\ units with a one-dimensional observed covariate $X\sim\mathcal N(0,1)$ and high-dimensional auxiliary covariates $Z\in\mathbb R^{p_Z}$ with $Z \sim \mathcal N(0,I_{p_Z}),$ $\|a\|_2=1$, where only the first $s\le p_Z$ coordinates of $a$ are nonzero (a sparse signal; in our experiments we take $s=\min(5,p_Z)$). Potential outcomes follow
\[
Y(0)=b_0 X + \gamma a^\top Z + \varepsilon_0,\qquad
Y(1)=b_1 X + \gamma a^\top Z + \varepsilon_1,
\]
where $b_0,b_1\sim \mathrm{Unif}(0.5,1.5)$ are drawn once per replication and $(\varepsilon_0,\varepsilon_1)$ are standard Gaussian with
\[
\mathrm{cor}(\varepsilon_0,\varepsilon_1)=\rho_\varepsilon\in[-1,1].
\]
Treatment is randomized, $T\sim\mathrm{Bernoulli}(1/2)$, and the observed outcome is $Y = T Y(1) + (1-T)Y(0)$.

For each replication, we compute the ``ground truth'' function $\rho(x)=\mathrm{cor}(Y(1),Y(0)\mid X=x)$ via kernel smoothing using the full simulated pairs $(Y_i(0),Y_i(1))$:
\[
\widehat{\rho}_{\mathrm{true}}(x)
=
\frac{\widehat{\mathrm{cov}}_w\{Y(0),Y(1)\}}
{\sqrt{\widehat{\mathrm{var}}_w\{Y(0)\}\,\widehat{\mathrm{var}}_w\{Y(1)\}}},
\]
where the weights are Gaussian in $X$ with bandwidth $h_{\mathrm{truth}}$. We evaluate on a grid $x\in\{q_{0.05}(X),\dots,q_{0.95}(X)\}$. We then apply our estimators $\widehat\rho_L(x)$ and  $\widehat{\tilde\rho}_L(x)$ using either random forest or linear regression as conditional mean estimators. 

The main summaries reported for each configuration are the average gaps
\( \frac{1}{|{\cal X}_{\mathrm{eval}}|}\sum_{x\in{\cal X}_{\mathrm{eval}}}
\bigl(\rho(x)-\widehat\rho_L(x)\bigr),\) and \(
 \frac{1}{|{\cal X}_{\mathrm{eval}}|}\sum_{x\in{\cal X}_{\mathrm{eval}}}
\bigl(\rho(x)-\widehat{\tilde\rho}_L(x)\bigr).
\)

We repeat the experiment for $p_Z\in\{5,20,100\}$ with $s=\min(5,p_Z)$, sample size $n=1000$, and $R=100$ Monte Carlo replications per configuration. For each configuration, we report Monte Carlo means with pointwise uncertainty bands computed as empirical quantiles across replications with $\alpha=0.05$. These are not asymptotic confidence intervals for a population parameter, but rather Monte Carlo error bands summarizing variability across replications for fixed $(n,p_Z,\gamma,\rho_\varepsilon)$. The results are drawn in Plot~\ref{fig_high_dim_rho_est}.

\textbf{Conclusion}: Overall, the qualitative behavior is very similar to the low-dimensional setting in Appendix~\ref{appendix_rho_estimation}. Increasing $\dim(Z)$ primarily inflates the variability, as expected from the more challenging nuisance estimation problem. Importantly, even at $\dim(Z)=100$ the estimated gaps and their uncertainty bands remain stable across $\gamma$ and $\rho_\varepsilon$ and continue to match the theoretical ordering of the bounds. This suggests that the proposed procedure is reasonably robust to moderate dimensions of auxiliary covariates in this sparse-signal regime.

\begin{figure}
    \centering
\includegraphics[width=0.95\linewidth]{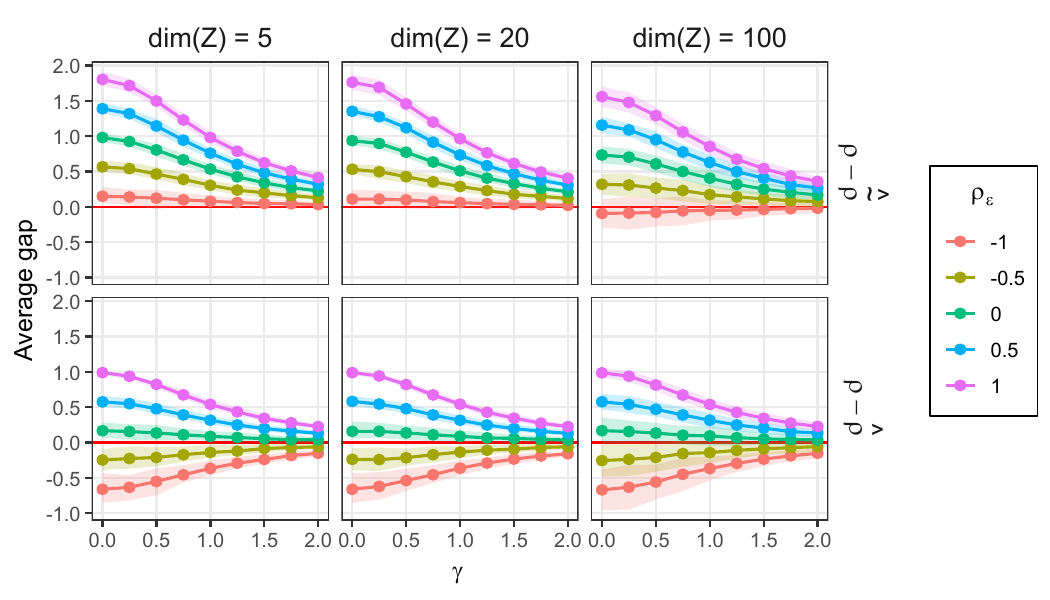}
 \caption{Average gap between the true cross-world correlation and its estimator as a function of the predictive power of the auxiliary variable $\gamma$ and dimension of the auxiliary variable $Z$. The upper panels report $\mathrm{GAP}(\gamma)=\mathbb{E}_X\!\left[\rho(X)-\widehat{\tilde\rho}(X)\right]$, while the lower panels report $\mathrm{GAP}(\gamma)=\mathbb{E}_X\!\left[\rho(X)-\widehat{\rho}(X)\right]$. Shaded bands are $95\%$ confidence intervals; the horizontal line at $0$ indicates no bias.}
    \label{fig_high_dim_rho_est}
\end{figure}

\section{Appendix: additional theoretical guarantees}
\label{section_appendix_theoretical_guarantees}

This appendix expands on the theoretical properties of the \( CW(\rho) \) prediction intervals, focusing on two important edge cases: perfect negative correlation (\( \rho = -1 \)) and perfect positive correlation (\( \rho = 1 \)) between potential outcomes. We present two theorems. Theorem~\ref{Theorem_sum} provides sufficient conditions for marginal coverage in the worst-case setting, where the counterfactuals are perfectly negatively correlated. Theorem~\ref{rho_1_theorem} shows that even assuming perfect alignment (\( \rho = 1 \)) yields valid coverage bounds under mild assumptions on bias of CATE estimator.

\begin{theorem}[Largest-width Scenario $\rho = -1$]
\label{Theorem_sum}
Assume that $\rho = \operatorname{cor}\big(Y(0), Y(1) \mid X = x\big) = -1$ for all $x\in\mathcal{X}$.  Suppose we have found two functions $\tilde l_0(x) = \hat{\mu}_0(x) -l_0(x)$ and $\tilde u_1(x)=\hat{\mu}_1(x) +u_1(x)$ that provide marginally valid prediction intervals:
\begin{equation}
    \label{eq1_in_thm1}
    \mathbb{P}\big(Y(0) \geq \tilde l_0(X)\big) \geq 0.9\quad \text{ and }\quad  \mathbb{P}\big(Y(1) \leq \tilde u_1(X)\big) \geq 0.9.
\end{equation}
Then, the $CW(\rho)$ prediction intervals are marginally valid:
\[
\mathbb{P}\bigg(Y(1) - Y(0) \leq   \hat{\tau}(X)  + D_{\rho}\big(l_0(X), u_1(X)\big)\bigg) \geq 0.9,
\]
if at least one of the following conditions is satisfied:
\begin{enumerate}
    \item \textbf{Conditional Coverage:} The bounds satisfy conditional coverage:    \[
 \mathbb{P}\big(Y(0) \geq \tilde l_0(X) \mid X = x\big) \geq 0.9 \quad \text{ and } \quad    \mathbb{P}\big(Y(1) \leq \tilde u_1(X) \mid X = x\big) \geq 0.9, \quad \forall x \in \mathcal{X}.
    \]
    \item \textbf{Similarity:} The bounds have similar coverage:
\[
    \mathbb{P}\big(Y(0) \geq \tilde l_0(X) \mid X = x\big) = \mathbb{P}\big(Y(1) \leq \tilde u_1(X) \mid X = x\big), \quad \forall x \in \mathcal{X}.
    \]
    \item \textbf{Naive:} $\mathbb{P}\big(Y(0) \geq \tilde l_0(X)\big) \geq 0.95$ and $\mathbb{P}\big(Y(1) \leq \tilde u_1(X)\big) \geq 0.95$.
    \item \textbf{Gaussianity/Convexity of the Quantile Function:} The quantile function of $Y(t) \mid X = x$ is convex on $[\delta_t,1]$, where
    \[
    \delta_0 =\inf_{x\in\mathcal{X}} \mathbb{P}\big(Y(0) \geq \tilde l_0(X) \mid X = x\big),  \quad  \delta_1 = \inf_{x\in\mathcal{X}} \mathbb{P}\big(Y(1) \leq \tilde u_1(X) \mid X = x\big).
    \]
    (The quantile function is typically convex in its tail for most standard distributions, such as Uniform, Gaussian Gamma, etc.).
\end{enumerate}
Additionally, the bound is tight in the sense that, without additional assumptions—even if all of the above conditions hold—there exists a situation where for any $\epsilon>0, $
\[
\mathbb{P}\bigg(Y(1) - Y(0) \leq \hat{\tau}(X)  + D_{\rho}\big(l_0(X), u_1(X)\big)  - \epsilon\bigg) < 0.9.
\]

\end{theorem}

Theorem~\ref{Theorem_sum} shows that strong finite-sample coverage guarantees can be achieved by imposing mild assumptions, either on the quality of individual predictions (such as conditional validity or similar prediction errors across treatment groups) or on the distribution of potential outcomes (for example, assuming Gaussianity). As an alternative, one can calibrate separate prediction intervals \eqref{separate_def} at the $95\%$ level. In the worst-case scenario, where  \( \rho = -1 \) and an ``ugly'' distribution meets unreliable conditional coverage, this set can be tight (see Example~\ref{example_for_ultra_tighness} in Appendix~\ref{appendix_proofs} for a case where the naive prediction set is the smallest valid set). Notably, in this case, our proposed prediction intervals are equivalent to the naive construction from \cite{Lei_2021_Conformal_Inference}, since  
\(
\hat{\tau}(X) + D_{-1}\big(l_0(X), u_1(X)\big) = \tilde{u}_1(X) - \tilde{l}_0(X)
\). 

Note that all three conditions in Theorem~\ref{Theorem_sum} are testable, as they involve only treated or untreated units separately and do not rely on any cross-world assumptions. For example, Conditions 1 and 2 can be evaluated using a standard train/test split, as commonly done in the conformal inference literature. Condition 3 can be satisfied by construction, though this may result in wider prediction intervals. Condition 4 can potentially be assessed using normality tests on the residuals.

The following theorem suggests that, unless the estimator $\hat \tau$ is highly biased, assuming $\rho = 1$ is sufficient to guarantee coverage of the $CW(\rho)$-bounds. In what follows, we distinguish between 'ground truth' $\rho$ and 'used hyperparameter' $\tilde \rho$ which can potentially differ. Note that for any $\tilde\rho \leq \rho$, $CW(\tilde\rho)$ bounds are wider than $CW(\rho)$ bounds. 

\begin{theorem}[Assuming $\rho = 1$ is sufficient]
\label{rho_1_theorem}
Assume that $\rho = \operatorname{cor}\big(Y(0), Y(1)\mid X = x\big) = 1$. Suppose that we found marginally valid intervals:
\[
\mathbb{P}\big(Y(0) \geq \hat{\mu}_0(X)-l_0(X)\big) \geq 0.9, \quad \mathbb{P}\big(Y(1) \leq \hat{\mu}_1(X)+u_1(X)\big) \geq 0.9.
\]
Then, for any $\tilde\rho\in[-1,1]$, the $CW(\tilde\rho)$-bounds are marginally valid:
\[
\mathbb{P}\bigg(Y(1)- Y(0) \leq  \hat{\tau}(X)  + D_{\tilde\rho}\big(l_0(X), u_1(X)\big)\bigg) \geq 0.9,
\]
provided that the \textbf{bias in $\boldsymbol{\hat{\tau}(x)}$ is small}, in the following sense:
\begin{equation*}
\begin{aligned}
 \text{if } b_x \leq  1: \quad & \tau(x)-\hat\tau(x) \leq D_{\tilde\rho}\big(l_0(x), u_1(x)\big) + (1 - b_x)\big(e_1(x)-u_1(x)\big), \\ \text{if } b_x \geq  1: \quad & \tau(x) -\hat\tau(x)\leq D_{\tilde\rho}\big(l_0(x), u_1(x)\big) + \big(1 - \frac{1}{b_x} \big) \big(e_0(x) - l_0(x)\big), 
\end{aligned}
\end{equation*}
where $e_1(x) =   \mu_1(x) - \hat  \mu_1(x) , e_0(x) =   -\mu_0(x) + \hat  \mu_0(x) $, and $b_x = \frac{\operatorname{sd}(Y(1) \mid X = x)}{\operatorname{sd}(Y(0) \mid X = x)}$, for $x\in\mathcal{X}.$ 
\end{theorem}

Theorem~\ref{rho_1_theorem} suggests that when $\rho = 1$, our proposed prediction intervals are marginally valid assuming some precision of the estimators. We discuss this in more details in Section~\ref{section_CI}. 

While Theorems \ref{Theorem_sum} and \ref{rho_1_theorem} focus on marginal coverage, a similar result can be also stated for conditional coverage: given \[
\mathbb{P}\big(Y(0) \geq \hat{\mu}_0(X)-l_0(X)\mid X=x \big) \geq 0.9, \quad \mathbb{P}\big(Y(1) \leq \hat{\mu}_1(X)+u_1(X)\mid X=x\big) \geq 0.9,
\]
we can similarly obtain conditionally valid 
\[
\mathbb{P}\bigg(Y(1)- Y(0) \leq  \hat{\tau}(X)  + D_{\tilde\rho}\big(l_0(X), u_1(X)\big)\mid X=x\bigg) \geq 0.9,
\]
given the assumptions in Theorems \ref{Theorem_sum} and \ref{rho_1_theorem}.

\section{Details on Numerical Experiments}
\label{appendix_simulations}
\subsection{Prediction intervals for ITE - literature review}
\label{lit_review}

\cite{Lei_2021_Conformal_Inference} proposed a naive approach that constructs \( C_0 \) and \( C_1 \) at level \( 1 - \alpha/2 \) as in \eqref{C_1_and_C_0_definitions}, and defines the ITE prediction set as  
\begin{equation}
    \label{Naive_definition}
    C_{\text{ITE}}^{\text{naive}} = C_1 + (-C_0) := \{c_1 - c_0 : c_1 \in C_1,\, c_0 \in C_0\}.
\end{equation}
They show that if \( C_0 \) and \( C_1 \) are correctly calibrated at level \( 1 - \alpha/2 \), then \( C_{\text{ITE}}^{\text{naive}} \) achieves marginal coverage at level \( 1 - \alpha \) as defined in \eqref{marginal_coverage_definition}.  Independently, \cite{Kivaranovic2020} showed that under independence of the potential outcomes (i.e., \( \rho = 0 \) in our notation), it suffices to calibrate \( C_0 \) and \( C_1 \) at level \( \sqrt{1 - \alpha} \) instead of \( 1 - \alpha/2 \) to ensure validity of \( C_{\text{ITE}}^{\text{naive}} \). While naive method enjoys strong theoretical guarantees, the resulting intervals are extremely conservative, often several times wider than optimal.

To address this, several strategies have been proposed to construct narrower intervals. In Section 4.2 of \cite{Lei_2021_Conformal_Inference}, a heuristic nested approach is suggested to shorten the intervals, without any coverage guarantees. \cite{Alaa_2023_Conformal_Meta_learners} introduce conformal meta-learners, which apply standard conformal prediction to pseudo-outcomes estimated via a chosen meta-learner.  \cite{jonkers2024conformalconvolutionmontecarlo} propose a Monte Carlo-based method that (assuming \( \rho = 0 \)) samples from estimated conditional distributions \( Y(1) \mid X \) and \( Y(0) \mid X \), and forms a pseudo-sample of ITEs by subtraction. Recent work by \citet{wang2025conformalinferenceindividualtreatment} uses conditional density estimation to reweight CQR under distributional shift between treated and untreated units, yielding improved intervals over those of \citet{Lei_2021_Conformal_Inference} in such settings.

In conclusion, existing approaches either offer coverage guarantees at the expense of overly wide prediction intervals, or produce narrower intervals without theoretical guarantees—though some perform reasonably well empirically when \( \rho = 0 \) (see Section~\ref{section_simulations}).

\subsection{Data}
\label{appendix_simulations_data}
We consider the following data-generating processes: 
\begin{itemize}
    \item \textbf{Synthetic}: if $d=1$ we generate covariate $X\sim Unif(-1,1)$. If $d>1$, we follow \cite{Wager_2018_Estimation_Inference, Alaa_2023_Conformal_Meta_learners, Lei_2021_Conformal_Inference, jonkers2024conformalconvolutionmontecarlo}, and we generate covariate vector \(
\mathbf{X} = (X_1, \ldots, X_d),
\)
where each covariate $X_j = \Phi(\tilde{X}_j)$, with $\Phi$ denoting the cumulative distribution function (CDF) of the standard normal distribution. The vector $(\tilde{X}_1, \ldots, \tilde{X}_d)$ follows an multivariate Gaussian distribution with mean zero and $\operatorname{Cov}(\tilde{X}_j, \tilde{X}_{j'}) = 0.25$ for $j \neq j'$. 

The treatment $T_i$ is generated via propensity score $\pi(\textbf{X}) = \frac{1 +|X_1|}{4}\in [0.25, 0.5]$, thereby providing sufficient overlap. The outcome is generated as \begin{align*}
Y_i(0) = f_0(\textbf{X}_i) + \varepsilon_i^0 \quad
Y_i(1) = f_0(\textbf{X}_i) + \tau(\textbf{X}_i)+ \varepsilon_i^1, \quad\text{where } (\varepsilon_i^0, \varepsilon_i^1)\sim \mathcal{N} \left(
\begin{bmatrix}
0 \\
0
\end{bmatrix},
\begin{bmatrix}
1 & 2\rho \\
2\rho & 4
\end{bmatrix}
\right),
\end{align*}
where $\tau(\textbf{x}) = \tau(x_1, x_2)$ is a random smooth polynomial depending only on the first two covariates (if $d=1$ only on the first covariate), generated using Perlin noise generator \citep{PerlinNoise} as in \cite{bodik2024structuralrestrictionslocalcausal}, and $f_0(x) = \beta^Tx$, where $\beta$ is generated via standard normal distribution.  
\item \textbf{IHDP (semi-synthetic):} Introduced in \cite{Bayes_Hill}, the dataset includes 25 observed pre-treatment covariates (\textbf{X}) such as birth weight, maternal age, educational background, etc. The treatment variable ($T$) indicates whether an infant received the intervention program (treated) or not (control). For each individual, the two potential outcomes ($Y_i(1), Y_i(0)$) correspond to cognitive test scores.  \cite{Bayes_Hill} simulated the potential outcomes as follows:
\begin{align}\label{eqIHDP}
Y_i(0) = f_0(X_i) + \varepsilon_i^0 \quad
Y_i(1) = f_1(X_i) + \varepsilon_i^1, \quad\text{where } \varepsilon_i^0, \varepsilon_i^1\overset{i.i.d.}{\sim} N(0,1).
\end{align}
Here, $f_0$ and $f_1$ are either linear functions (case ``A'') or more complex nonlinear functions (case ``B''). While this dataset is useful for benchmarking ATE and CATE estimators, the original implementation assumes $\rho = 0$, which may not reflect real-world behavior. Hence, we also considered the same data-generating process but using instead $$\begin{pmatrix}
\varepsilon_i^0 \\
\varepsilon_i^1
\end{pmatrix}
\sim \mathcal{N} \left(
\begin{pmatrix}
0 \\
0
\end{pmatrix},
\begin{pmatrix}
1 & \rho \\
\rho & 1
\end{pmatrix}
\right).$$ 
\item \textbf{Synthetic with non-Gaussian copula-based noise:} We use the same setup as in the first synthetic case but replace the Gaussian noise with samples from copula models combining different marginals and dependencies. Specifically, we consider Gaussian, Laplace, and \(t_3\) marginals, paired with Gaussian, Frank and $t$-copulas. 
\end{itemize}

\subsection{Additional experiments: IHDP dataset, non-Gaussian noise case and quantification of improvements of $CW(\rho)$}
\label{appendix_simulations_additional_experiments}
\begin{itemize}
    \item Figure~\ref{figure_result_IHDP} shows the coverage/width of various estimators applied on the IHDP dataset (with counterfactuals generated using a range of $\rho$). The results are similar to the case $d=15$ in Figure~\ref{fig_sim1}, although with larger variance.
    \item Figure~\ref{figure_result_copula} displays the empirical coverage of our \( CW(\rho) \) and \( CW^{+CI}(\rho) \) intervals on a synthetic dataset under various joint distributions of the noise. Each column corresponds to a different combination of copula and marginal distribution. We omit other copula–margin pairs, as they yielded similar results. For the combinations considered, the distribution of the noise appears to have little impact on the algorithm’s performance. While this simulation study represents only a single example and should not be interpreted as a conclusive proof, it hints some degree of robustness to deviations from Gaussian noise assumptions.
    \item Figure~\ref{figure_result_coverage_width_loss} compares the performance of prediction intervals of different methods across values of $\rho$ using the Coverage--Width loss, defined as
  \[
  \text{Coverage--Width Loss} = \frac{\text{width} - \text{width}_{\min}}{\text{width}_{\max} - \text{width}_{\min}} + \frac{2}{\alpha} \left| \text{coverage} - (1 - \alpha) \right|,
  \]
  where width is the average interval width across data points, and $\text{width}_{\min}$ and $\text{width}_{\max}$ are the minimum and maximum widths across methods for each combination of $\rho$ and dimension $d$. The Coverage--Width loss penalizes both wide intervals and miscalibrated coverage. 
  
  The results shown are aggregated across all experimental setups, including both IHDP and synthetic data described in Appendix~\ref{appendix_simulations_data} .

$CW(\rho)$ intervals achieve the best performance among all methods when $\rho \geq 0.5$, and remain competitive even when $\rho$ is misspecified. As theory suggests, $CW(\rho)$ becomes unstable near $\rho = 1$, but the enhanced variant $CW^{+CI}(\rho)$ maintains stable, state-of-the-art performance across all values of $\rho$. Overall, our CW-based intervals achieve 3--5× lower loss than other methods when $|\rho| \approx 1$, while for $\rho \approx 0$, $CW(\rho)$, CMC, and DR perform very similarly.
\end{itemize}

\begin{figure}
    \centering
    \includegraphics[width=1.05\linewidth]{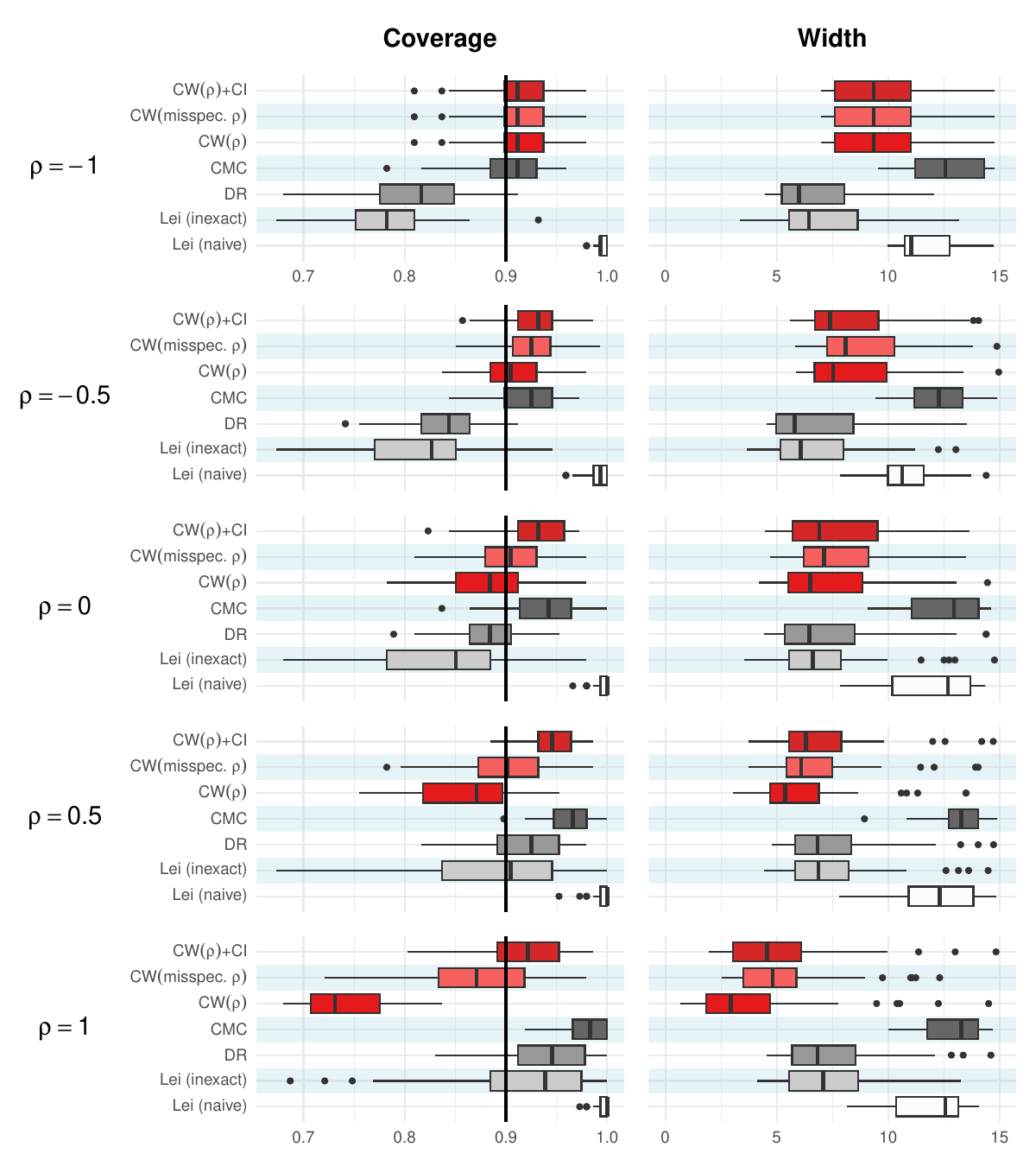} 
    \caption{Comparison of the coverage and average interval width of different interval estimation methods across varying correlation levels $\rho \in \{-1, -0.5, 0, 0.5, 1\}$ applied on the \textbf{IHDP} dataset. Methods include Lei's exact and inexact \citep{Lei_2021_Conformal_Inference}, conformal DR-meta-learner \citep{Alaa_2023_Conformal_Meta_learners}, CMC \citep{jonkers2024conformalconvolutionmontecarlo}, and two versions of our $CW(\rho)$ prediction intervals. The vertical line at 0.9 in the coverage panels marks the target coverage level.}
    \label{figure_result_IHDP}
\end{figure}

\begin{figure}
    \centering
    \includegraphics[width=1.05\linewidth]{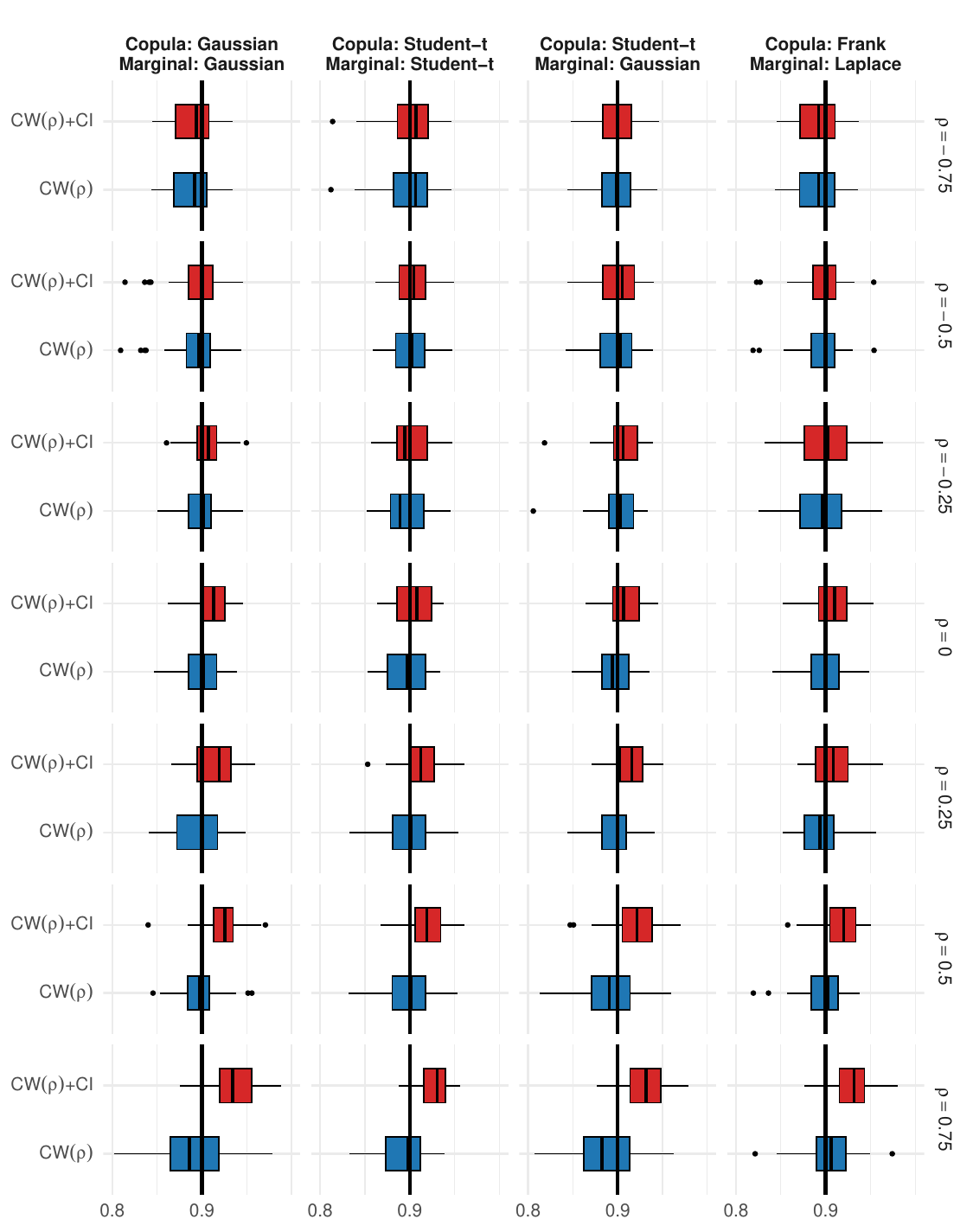} 
    \caption{Coverage across varying correlation levels $\rho$ applied on the synthetic dataset with different joint distribution of the noise. The vertical line at 0.9 in the coverage panels marks the target coverage level.}
    \label{figure_result_copula}
\end{figure}

\begin{figure}
    \centering
    \includegraphics[width=0.95\linewidth]{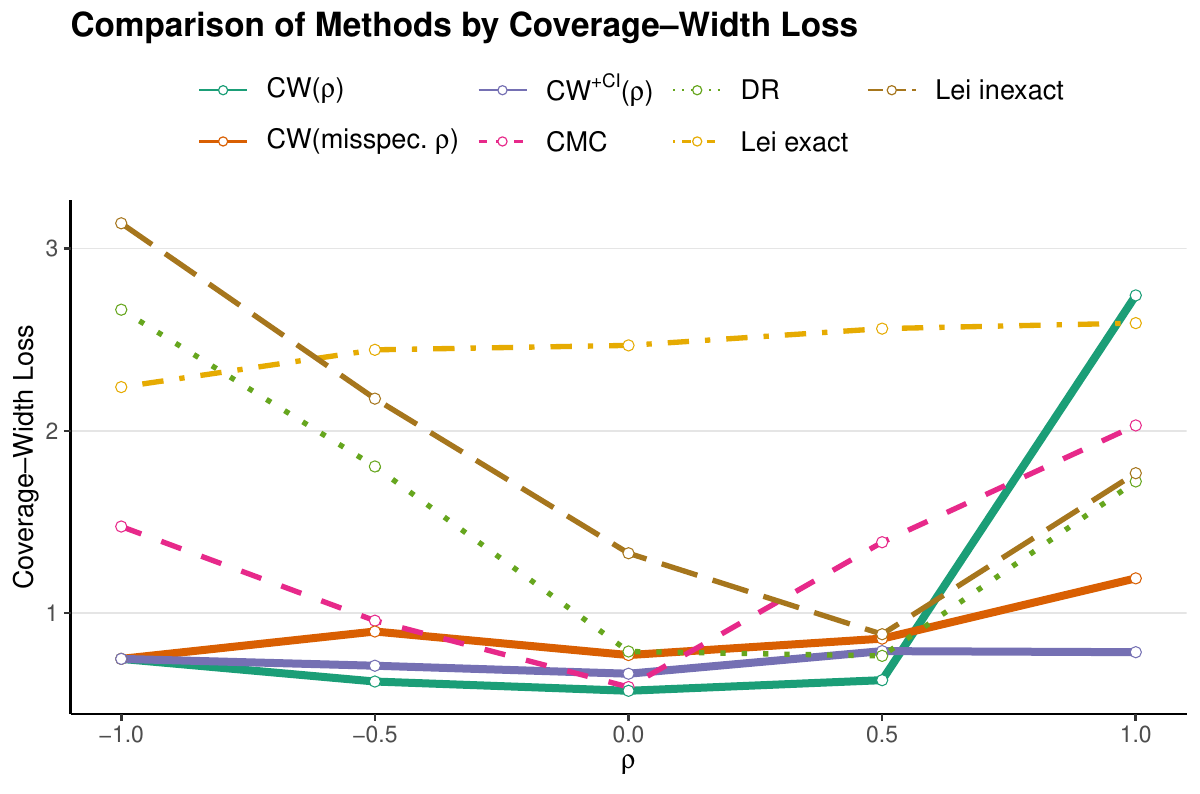} 
    \caption{Coverage--Width loss across values of $\rho$ for different prediction interval methods. Lower values indicate better performance in terms of coverage and interval width.}
    \label{figure_result_coverage_width_loss}
\end{figure}

\newpage
\section{Proofs}
\label{appendix_proofs}

\begin{customlem}{\ref{lemma_rho_L}}
   Assume that the relevant conditional second moments exist. Denote   $\sigma_t(x):=\mathrm{Var}\{Y(t)\mid X=x\}$, $\sigma_t(x,z):=\mathrm{Var}\{Y(t)\mid X=x,Z=z\},$ and $B_L(x):=\mathbb{E}\!\left[\sqrt{\sigma_0(x,Z)\,\sigma_1(x,Z)}\,\middle|\, X=x\right]$, $ t=0,1$. Then $B(x)\geq - B_L(x)$ and 
\begin{align}
\rho(x)\ \ge\ \tilde\rho_L(x)
:= \max\!\left\{
-1,\,
\frac{A(x)-B_L(x)}{\sqrt{\sigma_0(x)\sigma_1(x)}}
\right\}.
\tag{\ref{eq:rho_lower}}
\end{align}
 Additionally, if we assume $\mathrm{Cov}\{Y(0),Y(1)\mid X,Z\}\ge 0$, then the bound tightens to
\begin{equation}
\tag{\ref{eq_rho_L}}
\rho(x)\ge \rho_L(x):= A(x)/ \sqrt{\sigma_0(x)\sigma_1(x)}.
\end{equation}
\end{customlem}

\begin{proof}
\textbf{Step 1: a lower bound on $B(x)$.}
For any $z$ with $\mathbb{P}(Z=z\mid X=x)>0$, the conditional Cauchy--Schwarz inequality yields
\begin{align*}
\big|\mathrm{Cov}\{Y(0),Y(1)\mid X=x,Z=z\}\big|
&\le
\sqrt{\mathrm{Var}\{Y(0)\mid X=x,Z=z\}\,\mathrm{Var}\{Y(1)\mid X=x,Z=z\}} \\
&=
\sqrt{\sigma_0(x,z)\,\sigma_1(x,z)}.
\end{align*}
Hence,
\[
\mathrm{Cov}\{Y(0),Y(1)\mid X=x,Z=z\}\ \ge\ -\sqrt{\sigma_0(x,z)\,\sigma_1(x,z)}.
\]
Taking conditional expectation given $X=x$ on both sides and using the definition
\[
B(x):=\mathbb{E}\!\left[\mathrm{Cov}\{Y(0),Y(1)\mid X=x,Z\}\,\middle|\, X=x\right],
\]
we obtain
\[
B(x)\ \ge\ -\,\mathbb{E}\!\left[\sqrt{\sigma_0(x,Z)\,\sigma_1(x,Z)}\,\middle|\, X=x\right]
\]
\textbf{Step 2: a lower bound on $\rho(x)$.}
By the law of total covariance (conditional on $X=x$),
\[
\mathrm{Cov}\{Y(0),Y(1)\mid X=x\}=A(x)+B(x),
\]
where $A(x):=\mathrm{Cov}\{\mu_0(x,Z),\mu_1(x,Z)\mid X=x\}$ and
$\mu_t(x,z):=\mathbb{E}\{Y(t)\mid X=x,Z=z\}$.
Combining this decomposition with the bound from Step 1 gives
\[
\mathrm{Cov}\{Y(0),Y(1)\mid X=x\}\ \ge\ A(x)-B_L(x).
\]
Recalling that
\[
\rho(x):=\mathrm{Corr}\{Y(0),Y(1)\mid X=x\}
=\frac{\mathrm{Cov}\{Y(0),Y(1)\mid X=x\}}{\sqrt{\sigma_0(x)\sigma_1(x)}},
\]
we conclude that
\[
\rho(x)\ \ge\ \frac{A(x)-B_L(x)}{\sqrt{\sigma_0(x)\sigma_1(x)}}.
\]
Since any correlation is bounded below by $-1$, we may strengthen this to
\[
\rho(x)\ \ge\ \max\!\left\{-1,\,
\frac{A(x)-B_L(x)}{\sqrt{\sigma_0(x)\sigma_1(x)}}\right\}
=:\tilde\rho_L(x).
\]

\textbf{Step 3: tightened bound under nonnegative cross-world covariance.}
If, in addition, $\mathrm{Cov}\{Y(0),Y(1)\mid X,Z\}\ge 0$ almost surely, then
\[
B(x)=\mathbb{E}\!\left[\mathrm{Cov}\{Y(0),Y(1)\mid X=x,Z\}\,\middle|\, X=x\right]\ge 0.
\]
Therefore,
\[
\mathrm{Cov}\{Y(0),Y(1)\mid X=x\}=A(x)+B(x)\ \ge\ A(x),
\]
and dividing by $\sqrt{\sigma_0(x)\sigma_1(x)}$ yields
\[
\rho(x)\ \ge\ \frac{A(x)}{\sqrt{\sigma_0(x)\sigma_1(x)}}=: \rho_L(x).
\]
This completes the proof.
\end{proof}

\begin{lemma}[Asymptotic validity of $\widehat{\rho}_L(x)$]
\label{lem_appendix_rho_lower_consistency}
Fix $x$ such that $\sigma_0(x)\sigma_1(x)>0$. Let $\rho_L(x)$ be defined as in \eqref{eq:rho_lower} and let
$\widehat{\rho}_L(x)$ be the plug-in estimator obtained by replacing $A(x)$, $B_L(x)$ (if used), and
$\sigma_t(x)$ by estimators $\widehat A(x)$, $\widehat B_L(x)$, and $\widehat\sigma_t(x)$ constructed from an
auxiliary sample on which $(X,Z)$ are observed. Assume that:
\begin{enumerate}
\item[(i)] (\emph{Consistency of nuisance estimators}) For $t\in\{0,1\}$, $\widehat\mu_t(x,z)\overset{P}{\to} \mu_t(x,z)$ for
almost every $z$, and $\widehat\sigma_t(x)\overset{P}{\to}\sigma_t(x)$.
If $\widehat B_L(x)$ is used, additionally assume that $\widehat\sigma_t(x,z)\overset{P}{\to}\sigma_t(x,z)$ for almost every $z$.
\item[(ii)] (\emph{Moment conditions / domination}) $\E\{\mu_t(x,Z)^2\mid X=x\}<\infty$ and
$\E\{\sigma_t(x,Z)\mid X=x\}<\infty$ for $t\in\{0,1\}$, and the corresponding plug-in quantities are uniformly dominated
so that conditional dominated convergence applies.
\end{enumerate}
Then $\widehat{\rho}_L(x)\overset{P}{\to} \rho_L(x)$ as $n\to\infty$. Consequently, for any $\epsilon>0$,
\begin{align}
\lim_{n\to\infty}\Pr\!\left(\rho(x)\ \ge\ \widehat{\rho}_L(x)-\epsilon\right)=1.
\label{eq:rho_lower_asymp_ineq}
\end{align}
If, in addition, $\Cov\{Y(0),Y(1)\mid X,Z\}\ge 0$, then the same conclusion holds for the estimator $\widehat{\rho}_L(x)=\widehat A(x)/\{\widehat\sigma_0(x)\widehat\sigma_1(x)\}$.
\end{lemma}

\begin{proof}
We first establish the population lower bound. By \eqref{eq:cov_decomp_lower},
\[
\Cov\{Y(0),Y(1)\mid X=x\}=A(x)+B(x),
\]
where $A(x)=\Cov\{\mu_0(x,Z),\mu_1(x,Z)\mid X=x\}$ and
$B(x)=\E\!\left[\Cov\{Y(0),Y(1)\mid X=x,Z\}\mid X=x\right]$.
For each $z$, Cauchy--Schwarz implies
\begin{equation*}
    \begin{split}
&        \Cov\{Y(0),Y(1)\mid X=x,Z=z\} \\&\ge
-\sqrt{Var\{Y(0)\mid X=x,Z=z\}Var\{Y(1)\mid X=x,Z=z\}}\\&
= -\sqrt{\sigma_0(x,z)\sigma_1(x,z)}.
    \end{split}
\end{equation*}
Taking conditional expectation over $Z\mid X=x$ yields $B(x)\ge -B_L(x)$. Therefore,
\[
\Cov\{Y(0),Y(1)\mid X=x\}\ge A(x)-B_L(x),
\]
and dividing by $\sigma_0(x)\sigma_1(x)>0$ gives $\rho(x)\ge\rho_L(x)$ as in \eqref{eq:rho_lower}.
Under the sign restriction $\Cov\{Y(0),Y(1)\mid X,Z\}\ge 0$, we have $B(x)\ge 0$ and hence
$\rho(x)\ge A(x)/\{\sigma_0(x)\sigma_1(x)\}$.

We next show $\widehat{\rho}_L(x)\overset{P}{\to}\rho_L(x)$. Under (i)--(ii), conditional dominated convergence gives
\[
\E\{\widehat\mu_t(x,Z)\mid X=x\}\to \E\{\mu_t(x,Z)\mid X=x\},
\qquad
\E\{\widehat\mu_t(x,Z)^2\mid X=x\}\to \E\{\mu_t(x,Z)^2\mid X=x\},
\]
and thus $\widehat A(x)\overset{P}{\to} A(x)$ for any plug-in estimator $\widehat A(x)$ based on conditional moments of
$\widehat\mu_0(x,Z)$ and $\widehat\mu_1(x,Z)$ given $X=x$ (e.g.\ via binning or kernel smoothing). If $\widehat B_L(x)$
is used, the same argument applied to $\sqrt{\widehat\sigma_0(x,Z)\widehat\sigma_1(x,Z)}$ implies
$\widehat B_L(x)\overset{P}{\to} B_L(x)$. By assumption, $\widehat\sigma_t(x)\overset{P}{\to}\sigma_t(x)$ for $t\in\{0,1\}$.
Therefore, by Slutsky's theorem,
\[
\widehat{\rho}_L(x)
=
\frac{\widehat A(x)-\widehat B_L(x)}{\widehat\sigma_0(x)\widehat\sigma_1(x)}
\ \overset{P}{\to}\
\frac{A(x)-B_L(x)}{\sigma_0(x)\sigma_1(x)}
=
\rho_L(x),
\]
with the obvious simplification when the sign restriction is imposed (i.e.\ $\widehat B_L(x)\equiv 0$).

Finally, since $\rho(x)\ge\rho_L(x)$ deterministically, for any $\epsilon>0$,
\[
\Pr\!\left(\rho(x)\ge \widehat{\rho}_L(x)-\epsilon\right)
\ge
\Pr\!\left(\rho_L(x)\ge \widehat{\rho}_L(x)-\epsilon\right)
=
\Pr\!\left(\widehat{\rho}_L(x)\le \rho_L(x)+\epsilon\right)
\to 1,
\]
where the last convergence follows from $\widehat{\rho}_L(x)\overset{P}{\to}\rho_L(x)$. This proves
\eqref{eq:rho_lower_asymp_ineq}.
\end{proof}

\begin{example}
\label{example_for_ultra_tighness}
We show an example when the Naive set \eqref{Naive_definition} is tight in the following sense: for any $\epsilon>0$ there exist an example when the set 
$$C = \left[  \hat{\tau}(x) + u^{0.05}_0(x) + u^{0.05+\epsilon}_1(x), 
                 \hat{\tau}(x) - l^{0.05}_0(x) - l^{0.05+\epsilon}_1(x) \right]$$
does not have $90\%$ marginal coverage. 

 Consider the case $Y(1) =  - Y(0)$ (that is, $\rho = -1$), where 
    $$Y(1) =
\begin{cases} 
-1, & \text{with probability } 0.05+\epsilon \\
0, & \text{with probability } 0.9-\epsilon\\
1, & \text{with probability } 0.05
\end{cases}
$$
Consider $C_1 = [-1,0]$ and $C_0 = [0, 1]$. One can see that $P(Y(1)\in C_1) = 0.95$ and  $P(Y(0)\in C_0) = 0.95-\epsilon$. However,  the combination $C = [-1, 1]$ satisfy 
$P(Y(1) - Y(0) \in C) = 0.9-\epsilon<0.9$. Therefore, calibrating $Y(0)$ and $Y(1)$ on any smaller level than $0.95$ leads to non-valid coverage. 
\end{example}

\begin{customthm}{\ref{motivation_theorem}}[Motivation under a perfect (asymptotic) scenario]
Let $x\in\mathcal{X}$ and $\operatorname{cor}(Y(0), Y(1) \mid X = x) = \rho\in[-1, 1]$. Assume a perfect scenario:  $(Y(1), Y(0))\mid X=x$ is Gaussian,  $\hat{\mu}_t(x) =  \mu_t(x)$ and suppose that we found conditionally valid prediction intervals:
\begin{align*}
 &\mathbb{P}(Y(t) \leq \hat{\mu}_t(x) + u_t(x)\mid X=x)  = 0.95, \quad \mathbb{P}(Y(t) \geq \hat{\mu}_t(x) - l_t(x)  \mid X=x)=0.95, \,\,\,\,t=0,1.
\end{align*}
Then,  $CW(\rho)$ prediction intervals \eqref{D_rho_intervals_definition} are optimal in a sense that it is the smallest prediction set satisfying:
\begin{align*}
 \mathbb{P}(Y(1) - Y(0)\in C_{ITE}(X)\mid X=x)= 0.9.
\end{align*}
\end{customthm}

\begin{proof}
First, we introduce some notation:
\begin{itemize}
    \item Let \( c := \Phi^{-1}(0.95) \approx 1.6449 \) denote the 0.95 quantile of a standard Gaussian random variable.
        \item Let \( \sigma_t^2(x) := \operatorname{Var}(Y(t) \mid X = x) \) denote the conditional variance.
\end{itemize}
Without loss of generality, we assume that \( Y(t) \mid X = x \) is centered and that \( \hat{\mu}_t(x) = 0 \), for both $t=0, 1$. Under this assumption, observe the following:
\begin{itemize}
    \item \( u_t(x) = \operatorname{Quantile}_{0.95}(Y(t) \mid X = x) \).
    \item Since \( Y(t) \mid X = x \) is Gaussian and centered, symmetry implies \( l_t(x) = u_t(x) \).
    \item Therefore, \( u_t(x) = c \cdot \sigma_t(x) \), by the standard form of the quantile function for a Gaussian distribution.
    \item $\operatorname{Var}(Y(1) -  Y(0) \mid X=x) = \sigma_0^2(x) + \sigma_1^2(x) - 2\rho \sigma_0(x) \sigma_1(x)$ as a classical form of a sum of correlated gaussian random variables. Therefore, 
$$ \operatorname{Quantile}_{0.95}(Y(1)- Y(0) \mid X=x) = c\sqrt{\sigma_0^2(x) + \sigma_1^2(x) - 2\rho \sigma_0(x) \sigma_1(x) }.$$
\end{itemize}
Putting all these results together: 
\begin{align*}
    &\mathbb{P}(Y(1)- Y(0) \leq D_\rho (u_1(x), l_0(x))\\& =\mathbb{P}(Y(1)- Y(0)  \leq \sqrt{ l_0^2(x) + u_1^2(x)-2\rho l_0(x) u_1(x)}\mid X=x)\\& =\mathbb{P}(Y(1)- Y(0)  \leq \sqrt{ u_0^2(x) + u_1^2(x)-2\rho u_0(x) u_1(x)}\mid X=x) \\&=  \mathbb{P}(Y(1)- Y(0) \leq c\sqrt{\sigma_0^2(x) + \sigma_1^2(x) - 2\rho \sigma_0(x) \sigma_1(x) })\mid X=x)\\&
    = \mathbb{P}(Y(1)- Y(0) \leq  \operatorname{Quantile}_{0.95}(Y(1)- Y(0) \mid X=x) \mid X=x ) = 0.95.
\end{align*}
Analogously 
\begin{align*}
    &\mathbb{P}(Y(1)- Y(0)  \geq -\sqrt{ u_0^2(x) + l_1^2(x)-2\rho u_0(x) l_1(x)}\mid X=x) = 0.95.
\end{align*}
Hence, we proved that \begin{align*}
 \mathbb{P}(Y(1) - Y(0)\in C_{ITE}(X)\mid X=x)= 0.9.
\end{align*}
Finally, it is well known that for any Gaussian random variable \( Z \), the shortest prediction interval with 90\% coverage is given by  
\(
\left( \operatorname{Quantile}_{0.05}(Z), \operatorname{Quantile}_{0.95}(Z) \right).
\)
By defining \( Z = Y(1) - Y(0) \) and conditioning on any \( X = x \), it follows that the \( D_\rho \) prediction set \( C_{\text{ITE}}(x) \) is the smallest possible interval achieving the desired coverage. Moreover, by continuity of the Gaussian distribution, for any \( \epsilon > 0 \), we necessarily have:
\[
\mathbb{P}\left(Y(1) - Y(0) \leq \sqrt{ l_0^2(x) + u_1^2(x) - 2\rho\, l_0(x) u_1(x)} - \epsilon \mid X = x \right) < 0.95,
\]
\[
\mathbb{P}\left(Y(1) - Y(0) \geq \sqrt{ u_0^2(x) + l_1^2(x) - 2\rho\, u_0(x) l_1(x)} - \epsilon \mid X = x \right) < 0.95.
\]
\end{proof}

\begin{customthm}{\ref{Theorem_sum}}[Largest-width Scenario $\rho = -1$]
\,
\begin{tcolorbox}[colback=gray!5!white, colframe=gray!20!white, boxrule=0.2pt, left=1pt, right=1pt]
\begin{flushleft}
For notational convenience in the proof, we re-define 
\[
\tilde{Y}_0 := -Y(0), \qquad \tilde{Y}_1 := Y(1), \quad \text{ and } \quad f_t(x) := \mathbb{E}[\tilde{Y}_t \mid X=x], \quad t=0,1,
\]
so that we are interested in $ITE = \tilde{Y}_1 + \tilde{Y}_0$ and CATE becomes $\tau(x) = f_0(x) + f_1(x)$. With this notation in place, an equivalent formulation of Theorem~\ref{Theorem_sum} is given below. 
\end{flushleft}
\end{tcolorbox}

Assume that $\rho = -1$, i.e., $\operatorname{cor}(\tilde{Y}_0, \tilde{Y}_1 \mid X = x) = +1$.  Suppose we have found two functions $\tilde u_0(x) = \hat{f}_0(x) + u_0(x)$ and $\tilde u_1(x)=\hat{f}_1(x) +u_1(x)$ that provide marginally valid upper bounds:
\[
\mathbb{P}(\tilde{Y}_t \leq \tilde u_t(X)) \geq 0.9, \quad t = 0,1.
\]
Then, the $CW(\rho)$ prediction intervals are marginally valid:
\[
\mathbb{P}(\tilde{Y}_1 + \tilde{Y}_0 \leq   \hat{\tau}(X)  + D_{\rho}\big(u_0(X), u_1(X)\big)) \geq 0.9,
\]
if at least one of the following conditions is satisfied:
\begin{enumerate}
    \item \textbf{Conditional Coverage:} The bounds satisfy conditional coverage:    \[
    \mathbb{P}(\tilde{Y}_t \leq \tilde u_t(X) \mid X = x) \geq 0.9, \quad \forall x \in \mathcal{X}, \quad t = 0,1.
    \]
    \item \textbf{Similarity:} The bounds have similar coverage:
\[
    \mathbb{P}(\tilde{Y}_0 \leq \tilde u_0(X) \mid X = x) = \mathbb{P}(\tilde{Y}_1 \leq \tilde u_1(X) \mid X = x), \quad \forall x \in \mathcal{X}.
    \]
    \item \textbf{Naive:} $\mathbb{P}(\tilde{Y}_t \leq \tilde u_t(X)) \geq 0.95$ for both $t=0,1$.
    \item \textbf{Convexity of the Quantile Function:} The quantile function of $\tilde{Y}_t \mid X = x$ is convex on $[\delta,1]$, where
    \[
    \delta = \inf_{x\in\mathcal{X}} \mathbb{P}(\tilde{Y}_t \leq \tilde u_t(X) \mid X = x), \quad t = 0,1.
    \]
    (The quantile function is typically convex in its tail for most standard distributions, such as Uniform, Gaussian Gamma, etc.).
    \item $\boldsymbol{\rho = +1}$ while we use $ D_{-1}$ prediction bounds.
\end{enumerate}
Additionally, the bound is tight in the sense that, without additional assumptions—even if all of the above conditions hold—there exists a situation where for any $\epsilon>0, $
\[
\mathbb{P}(\tilde{Y}_0 + \tilde{Y}_1 \leq \hat{\tau}(X)  + D_{\rho}\big(u_0(X), u_1(X)\big)  - \epsilon) < 0.9.
\]
\end{customthm}

\begin{proof}
Notice that $$ \hat{\tau}(X)  + D_{\rho}\big(u_0(X), u_1(X)\big) = \tilde u_0(X)+\tilde u_1(X) .$$
Define $$\phi_t(x) = \mathbb{P}(\tilde{Y}_t \leq \tilde u_t(X) \mid X = x), \quad t=0,1.$$ 

\textbf{Bullet-point 1: }  Direct consequence of  Lemma~\ref{lemma_about_rho_pm1}. In fact, this guarantees us not only marginal coverage, but even stronger conditional coverage. 

\textbf{Bullet-point 2: }  We have for $t=0,1$$$\int_x   \phi_t(x) dF_X(x) \geq 0.9.$$Using Lemma~\ref{lemma_about_rho_pm1}, for each $x\in\mathcal{X}$ we have $$\mathbb{P}(\tilde{Y}_0 +\tilde{Y}_1\leq \tilde u_0(X)+\tilde u_1(X) \mid X = x)\geq \min(\phi_0(x),\phi_1(x))  .$$
Together
\begin{equation*}
    \begin{split}
&\int_x   \mathbb{P}(\tilde{Y}_0+\tilde{Y}_1 \leq \tilde u_0(X)+\tilde u_1(X) \mid X = x) dF_X(x)        \\&
\geq  \int_x \min(\phi_0(x), \phi_1(x)) dF_X(x)\\& 
= \int_x \phi_t(x) dF_X(x) \geq 0.9.
    \end{split}
\end{equation*}

\textbf{Bullet-point 3: }  This statement follows from Bonferroni correction, as also noted in \citep{Lei_2021_Conformal_Inference}, since $\mathbb{P}(\tilde{Y}_0 \leq \tilde u_0(X)) \geq 0.95$ and $\mathbb{P}(\tilde{Y}_1\leq \tilde u_1(X)) \geq 0.95$ always implies $\mathbb{P}(\tilde{Y}_0 + \tilde{Y}_1 \leq \tilde u_0(X) + \tilde u_1(X)) \geq 0.9$.

\textbf{Bullet-point 4: }  
\textbf{Fact: }For any two random variables $Z_0, Z_1$ with finite correlation, it holds that $cor(Z_0, Z_1)=1\iff Z_1 = a+bZ_0$ for some $a,b\in\mathbb{R}$ (follows directly from Cauchy-Schwarz inequality). 

Fix $x\in\mathcal{X}$. Denote $\tilde{Y}_1 = a+b\tilde{Y}_0$ where $b>0$, conditioning on $X=x$. Due to convexity, we directly have (see also Lemma~\ref{lemma_number_46})
$$\mathbb{P}(\tilde{Y}_0+\tilde{Y}_1 \leq \tilde u_0(X)+\tilde u_1(X) \mid X = x)\geq \frac{\phi_0(x) + b\phi_1(x)}{1+b}.$$
Therefore
\begin{equation*}
    \begin{split}
&\int_x   \mathbb{P}(\tilde{Y}_0+\tilde{Y}_1 \leq \tilde u_0(X)+\tilde u_1(X) \mid X = x) dF_X(x)        \\&
\geq  \int_x \frac{\phi_0(x) + b\phi_1(x)}{1+b} dF_X(x)\\& 
= \frac{1}{1+b}\int_x \phi_0(x) dF_X(x) +  \frac{b}{1+b}\int_x \phi_1(x) dF_X(x)\\&\geq \frac{1}{1+b}0.9 +  \frac{b}{1+b}0.9= 0.9.
    \end{split}
\end{equation*}

\textbf{Bullet-point 5:} 
Fix $x\in\mathcal{X}$ and in the proof consider conditioning on $X=x$ in every step. Define \( \tilde{Y}^{cen}_t := \tilde{Y}_t - \mathbb{E}[\tilde{Y}_t] \) for \( t = 0, 1 \). It follows that \( \tilde{Y}^{cen}_1 = -b_x \tilde{Y}^{cen}_0 \) for some $b_x>0$. Assume $b_x\leq 1$, otherwise we switch $\tilde{Y}^{cen}_1$ with $\tilde{Y}^{cen}_0$. Simple computation gives us 
\begin{align*}
&   \mathbb{P}(\tilde{Y}_0 +\tilde{Y}_1\leq \tilde u_0(X)+\tilde u_1(X) \mid X = x) 
\\&= \mathbb{P}(\tilde{Y}^{cen}_0 +\tilde{Y}^{cen}_1 \leq \tilde u_0(x) - f_0(x) +\tilde u_1(x) - f_1(x)\mid X=x  )
 \\&= \mathbb{P}((1-b_x)\tilde{Y}^{cen}_0 \leq\tilde u_0(x) - f_0(x) + \tilde u_1(x) - f_1(x)\mid X=x  )
  \\&\geq \mathbb{P}(\tilde{Y}^{cen}_0 \leq  \tilde u_0(x) - f_0(x)\mid X=x  )
  \\& = \mathbb{P}( \tilde{Y}_0 \leq  \tilde u_0(x) \mid X=x  )
\end{align*}
Since this holds for every $x$, we have $ \mathbb{P}(\tilde{Y}_0 +\tilde{Y}_1\leq\tilde u_0(X)+\tilde u_1(X) ) \geq  \mathbb{P}( \tilde{Y}_0 \leq  \tilde u_0(X))\geq 0.9$. 

\textbf{Tightness: } Consider any continuous bounded distribution $F$ (possibly with a convex quantile function) with upper support bound $sup(F)<\infty$. Let 
$$X =
\begin{cases} 
-1, & \text{with probability } 0.05 \\
0, & \text{with probability } 0.05 \\
1, & \text{with probability } 0.9, 
\end{cases}\quad \quad
\tilde{Y}_1 = \begin{cases} 
\mathbb{P}(\tilde{Y}_1=10)= 1, & if\,\,\,X=-1 \\
\mathbb{P}(\tilde{Y}_1=10)= 1, & if\,\,\,X=0 \\
\tilde{Y}_1\sim F, & if\,\,\,X=1 , 
\end{cases}
$$

$$
u_0(x) = \begin{cases} 
10, & if\,\,\,x=-1 \\
9, & if\,\,\,x=0 \\
sup(F) & if\,\,\,x=1 , 
\end{cases}, 
u_1(x) = \begin{cases} 
9, & if\,\,\,x=-1 \\
10, & if\,\,\,x=0 \\
sup(F) & if\,\,\,x=1 , 
\end{cases}
$$
Then, $\mathbb{P}(\tilde{Y}_1 + \tilde{Y}_0 \leq\tilde u_0(X) + \tilde u_1(X))=0.9$ since this happens if and only if $X=1$. However, for any $\epsilon$ we have that $\mathbb{P}(\tilde{Y}_1 + \tilde{Y}_0 \leq\tilde u_0(X) + \tilde u_1(X)-\epsilon)=0.9F(sup(F) - \epsilon)<0.9$ 

\end{proof}

\begin{lemma}[Case $\rho =\pm 1$]\label{lemma_about_rho_pm1}
Let $X, Y$ be a pair of random variables such that $Y = a + b X$ for some $a,b\in\mathbb{R}$. Let $q_x, q_y\in\mathbb{R}$ such that $$\mathbb{P}(X\leq  q_x)\geq 0.9, \,\,\,\,\,\, \mathbb{P}(Y\leq q_y)\geq 0.9.$$
Then $$\mathbb{P}(X+Y \leq q_x + q_y)\geq 0.9.$$
\end{lemma}\begin{proof}
    If $b=0$ then the statement is trivial. If $b=-1$ then $X+Y = a$ and the statement is, again, trivial. From now on, assume $b\neq 0, -1. $

Define $q'_y = (q_y-a)/b$ and notice that $\mathbb{P}(X\leq q'_y)\geq 0.9$. Since $X+Y = a + (b+1)X$, rewrite $X+Y\leq q_x + q_y$ as 
\[
X+Y\leq q_x + q_y \iff X \leq \frac{q_x + q_y - a}{1 + b} =  \frac{q_x + bq'_y}{1 + b}.
\]
Therefore 
$$\mathbb{P}(X + Y \leq q_x + q_y) = \mathbb{P}(X\leq  \frac{q_x + bq'_y}{1 + b}) \geq  \mathbb{P}(X\leq \min\{q_x, q'_y\})\geq 0.9.$$

\end{proof}

\begin{lemma}
\label{lemma_number_46}
Let $X,Y$ be random variables such that $Y = a+bX$ for some $a\in\mathbb{R},b>0$. Define $Q_X(\cdot), Q_Y(\cdot)$ quantile function of $X$ and $Y$. If $Q_X$ is convex on $(\beta_0, \beta_1)$, where  $\beta_0<\beta_1\in(0,1)$, then
 $$\mathbb{P}(X+Y\leq Q_X(\beta_0) +  Q_Y(\beta_1))\geq \frac{\beta_0 + b\beta_1}{1+b}. $$
    
\end{lemma}

\begin{proof}

Using Jensen inequality, it holds that:  
\begin{align*}
&\mathbb{P}(X+Y\leq Q_X(\beta_0) +  Q_Y(\beta_1))
=\mathbb{P}(a+(1+b)X\leq Q_X(\beta_0) + a + b Q_X(\beta_1))\\&
\mathbb{P}(X\leq \frac{1}{1+b}Q_X(\beta_0)  + \frac{b}{1+b} Q_X(\beta_1))
\geq \mathbb{P}(X\leq Q_X(\frac{\beta_0 + b\beta_1}{1+b}) ) \\&= \frac{\beta_0 + b\beta_1}{1+b}. 
\end{align*}
\end{proof}

\begin{customthm}{\ref{rho_1_theorem}}[Best-Case Scenario $\rho = 1$]
\,
\begin{tcolorbox}[colback=gray!5!white, colframe=gray!20!white, boxrule=0.2pt, left=1pt, right=1pt]
\begin{flushleft}
For notational convenience in the proof, we re-define 
\[
\tilde{Y}_0 := -Y(0), \qquad \tilde{Y}_1 := Y(1), \quad \text{ and } \quad f_t(x) := \mathbb{E}[\tilde{Y}_t \mid X=x], \quad t=0,1,
\]
so that we are interested in $ITE = \tilde{Y}_1 + \tilde{Y}_0$ and CATE becomes $\tau(x) = f_0(x) + f_1(x)$. With this notation in place, an equivalent formulation of Theorem~\ref{rho_1_theorem} is given below. 
\end{flushleft}
\end{tcolorbox}

Assume that $\rho = 1$, i.e., $\operatorname{cor}(\tilde{Y}_0, \tilde{Y}_1 \mid X = x) = -1$. Suppose that our bounds are marginally valid:
\[
\mathbb{P}(\tilde{Y}_t \leq \hat{f}_t(X)+u_t(X)) \geq 0.9, \quad t = 0,1.
\]
Then, for any $\tilde\rho\in[-1,1]$, the $D_{\tilde\rho}$-bounds are marginally valid:
\[
\mathbb{P}\bigg(\tilde{Y}_0 + \tilde{Y}_1 \leq  \hat{\tau}(X)  + D_{\tilde\rho}\big(u_0(X), u_1(X)\big)\bigg) \geq 0.9,
\]
provided that the bias in $\hat{\tau}(x)$ is small, in the following sense:
\begin{equation}\label{eq_extremly_underestimated}
\begin{aligned}
 \text{if } b_x \leq  1: \quad & \tau(x)-\hat\tau(x) \leq D_{\tilde\rho}\big(u_0(x), u_1(x)\big) + (1 - b_x)\big(e_1(x)-u_1(x)\big), \\ \text{if } b_x \geq  1: \quad & \tau(x) -\hat\tau(x)\leq D_{\tilde\rho}\big(u_0(x), u_1(x)\big) + \big(1 - \frac{1}{b_x} \big) \big(e_0(x) - u_0(x)\big), 
\end{aligned}
\end{equation}
where $e_t(x) =   f_t(x) - \hat f_t(x) $, $t=0,1$, and $b_x = \frac{\operatorname{sd}(\tilde{Y}_1 \mid X = x)}{\operatorname{sd}(\tilde{Y}_0 \mid X = x)}$, for $x\in\mathcal{X}.$ 
\end{customthm}

\begin{proof}
Fix \( x \in \mathcal{X} \). Define the function \( \phi_t(x) \) as:
\[
\phi_t(x) = \mathbb{P}\left( \tilde{Y}_t \leq \hat{f}_t(X) + u_t(X) \mid X = x \right), \quad t = 0, 1.
\]
We use the following well-known fact: For any two random variables \( Z_0 \) and \( Z_1 \) with finite correlation, it holds that \( \operatorname{cor}(Z_0, Z_1) = -1 \) if and only if \( Z_1 = a - b Z_0 \) for some constants \( a, b \in \mathbb{R} \), with \( b > 0 \) (this follows directly from the Cauchy-Schwarz inequality). Hence, conditioning on \( X = x \), we define the relation \( \tilde{Y}_1 = a_x - b_x \tilde{Y}_0 \), where \( b_x \) is given by
\[
b_x = \frac{\operatorname{sd}(\tilde{Y}_1 \mid X = x)}{\operatorname{sd}(\tilde{Y}_0 \mid X = x)}.
\]
Define the centered versions \(\tilde{Y}^{cen}_t := \tilde{Y}_t - \mathbb{E}[\tilde{Y}_t] \) for \( t = 0, 1 \). It follows that \(\tilde{Y}^{cen}_1 = -b_x\tilde{Y}^{cen}_0 \).

\textbf{Case \( b_x = 1 \):} In this case, \(\tilde{Y}^{cen}_1 +\tilde{Y}^{cen}_0 \overset{a.s.}{=} 0 \). Hence,
\[
\mathbb{P}\bigg(\tilde{Y}_0 + \tilde{Y}_1 \leq  \hat{\tau}(X)  + D_{\tilde\rho}\big(u_0(X), u_1(X)\big)\bigg) =\mathbb{P}\bigg(0\leq  e_0(x)+e_1(x)  + D_{\tilde\rho}\big(u_0(X), u_1(X)\big)\bigg).
\]
This is satisfied since \eqref{eq_extremly_underestimated}
reads as \[
e_0(x) + e_1(x) \leq  D_{\tilde\rho}\big(u_0(x), u_1(x)\big).
\]
This directly implies that the bounds hold with probability 1. 

\textbf{Main Argument:}
Consider the case $b_x<1$. Using simple computations,
\[
\mathbb{P}(\tilde{Y}_0+\tilde{Y}_1 \leq \hat{f}_0(x)+  \hat{f}_1(x) + D_\rho\big(u_0(x) , u_1(x)\big)  \mid X=x)
\]
can be rewritten as
\[
\mathbb{P}((1-b_x)\tilde{Y}^{cen}_0 \leq e_0(x) +   e_1(x) + D_\rho\big(u_0(x) , u_1(x)\big) \mid X=x ).
\]
Rearranging,
\[
\mathbb{P}(\tilde{Y}^{cen}_0 \leq \frac{e_0(x) +   e_1(x) + D_\rho\big(u_0(x) , u_1(x)\big)}{1-b_x} \mid X=x).
\]
Using the following algebraic fact
\begin{align*}
  &  e_0(x)+  b_xe_1(x) \leq D_{\rho}\big(u_0(x), u_1(x)\big) - (1-b_x)u_1(x),\\&\quad\quad\quad\quad\quad\quad\iff\\
    &\frac{e_0(x) +   e_1(x) + D_\rho\big(u_0(x) , u_1(x)\big)}{1-b_x} \geq  e_0(x)+u_0(x),
\end{align*}
we obtain 
\[
\mathbb{P}(\tilde{Y}^{cen}_0 \leq e_0(x)+u_0(x) \mid X=x) = \mathbb{P}(\tilde{Y}_0+\tilde{Y}_1 \leq  \hat\tau(x) + D_\rho\big(u_0(x) , u_1(x)\big)  \mid X=x).
\]
The marginal coverage property follows by integrating over \( X \). 

An analogous argument can be done when $b_x>1$ as the problem is symmetric. 
\end{proof}

\begin{customthm}{\ref{theorem_general_gaussian}}[General Gaussian case]
\,
\begin{tcolorbox}[colback=gray!5!white, colframe=gray!20!white, boxrule=0.2pt, left=1pt, right=1pt]
\begin{flushleft}
For notational convenience in the proof, we re-define 
\[
\tilde{Y}_0 := -Y(0), \qquad \tilde{Y}_1 := Y(1), \quad \text{ and } \quad f_t(x) := \mathbb{E}[\tilde{Y}_t \mid X=x], \quad t=0,1,
\]
so that we are interested in $ITE = \tilde{Y}_1 + \tilde{Y}_0$ and CATE becomes $\tau(x) = f_0(x) + f_1(x)$. With this notation in place, an equivalent formulation of Theorem~\ref{theorem_general_gaussian} is given below. 
\end{flushleft}
\end{tcolorbox}

Let $\rho\in[-1,1]$ and assume that $\operatorname{cor}(\tilde{Y}_0, \tilde{Y}_1 \mid X = x) = -\rho$. Suppose that the predictive bounds are marginally valid:
\[
\mathbb{P}(\tilde{Y}_t \leq \hat{f}_t(X)+u_t(X)) \geq 0.9, \quad t = 0,1.
\]
Then, for any $\tilde\rho\leq \rho$, the $CW(\tilde\rho)$ intervals remain marginally valid:
\[
\mathbb{P}\bigg(\tilde{Y}_0 + \tilde{Y}_1 \leq  \hat{\tau}(X)  + D_{\tilde\rho}\big(u_0(X), u_1(X)\big)\bigg) \geq 0.9,
\]
provided that the following conditions hold:  
\begin{itemize}
    \item \textbf{Gaussianity}: $(\tilde{Y}_0, \tilde{Y}_1)\mid X=x$ is normally distributed for all $x\in\mathcal{X}$, 
    \item \textbf{Coverage-variance relation}:   $\mathbb{E}[w(X)\phi_0(X)+ (1-w(X)) \phi_1(X)] \geq 0.9,$ where  
    \[
       \phi_t(x) := \mathbb{P}(\tilde{Y}_t \leq \hat{f}_t(X) + u_t(X) \mid X=x), \quad w(x) := (1-\rho^2) \frac{\operatorname{Var}(\tilde{Y}_0\mid X=x)}{\operatorname{Var}(\tilde{Y}_0+\tilde{Y}_1\mid X=x)} \in [0,1]. \]
    This condition holds, for example, under \textit{Conditional Coverage, Similarity, Homoskedasticity}, if  
    $\frac{\operatorname{Var}(\tilde{Y}_0\mid X=x)}{\operatorname{Var}(\tilde{Y}_1\mid X=x)} = \text{constant},$ or if $\rho = \pm 1$. 
    \item \textbf{$\boldsymbol{\hat\tau(x)}$ is not too biased}: $\tau(x) - \hat\tau(x) = e_0(x) +e_1(x) \leq  D_{\tilde\rho}\big(u_0(x), u_1(x)\big) - D_{\rho}\big(u_0(x)-e_0(x), u_1(x)-e_1(x)\big)$, where $e_t(x) = f_t(x)-\hat f_t(x)$. This holds for example if $\rho = -1$.
\end{itemize}
\end{customthm}

\begin{proof}

We use the notation $\tilde{u_t}:=u_t -e_t$. 

\textbf{Main idea: assuming $\rho = \tilde\rho = 0$ and  $\hat{f}_t(x) = f_t(x) = 0$:}
\begin{equation*}
    \begin{split}
&\mathbb{P}\bigg(\tilde{Y}_0 + \tilde{Y}_1 \leq  \hat{\tau}(X)  + D_{\tilde\rho}\big(u_0(X), u_1(X)\big)\bigg)\\&=\int_x   \mathbb{P}(\tilde{Y}_0+\tilde{Y}_1 \leq \sqrt{u_0^2(X)+u_1^2(X)} \mid X = x) dF_X(x)        \\&
\overset{Lemma~\ref{lemma_subGauss}, \text{ case }\rho = 0}{\geq } \int_x  w(x)\phi_0(x)+ (1-w(x))\phi_1(x) dF_X(x)\\& 
=  \mathbb{E}[w(X)\phi_0(X)+ (1-w(X)) \phi_1(X)] \geq 0.9.
    \end{split}
\end{equation*}

\textbf{General case:} 
\begin{equation*}
    \begin{split}
&\int_x   \mathbb{P}(\tilde{Y}_0+\tilde{Y}_1 \leq \hat\tau(x)+ \sqrt{u_0^2(x)+u_1^2(x) - 2\tilde\rho u_0(x)u_1(x)} \mid X = x) dF_X(x)        \\&
=\int_x   \mathbb{P}( \tilde{Y}_0 - f_0(x)+ \tilde{Y}_1 - f_0(x) \leq -e_0(x)-e_1(x) +  \sqrt{u_0^2(x)+u_1^2(x) - 2\tilde\rho u_0(x)u_1(x)} \mid X = x) dF_X(x)        \\&
\overset{Bias-Assumption}{\geq } \int_x   \mathbb{P}(\tilde{Y}_0 - f_0(x)+ \tilde{Y}_1 - f_0(x) \leq  D_\rho{(\tilde u_0(x),\tilde u_1(x))} \mid X = x) dF_X(x)        \\&
\overset{Lemma~\ref{lemma_subGauss}}{\geq } \int_x  w(x)\phi_0(x)+ (1-w(x))\phi_1(x) dF_X(x)=\mathbb{E}[w(X)\phi_0(X)+ (1-w(X)) \phi_1(X)]\\& \overset{Coverage-variance}{\geq } 0.9 .
    \end{split}
\end{equation*}

\textbf{Observation 1:} The condition  
\[
\mathbb{E}[w(X)\phi_0(X)+ (1-w(X)) \phi_1(X)] \geq 0.9
\]  
holds under \textit{Conditional Coverage, or Similarity, or Homoskedasticity}, or if  
\(
\frac{\operatorname{Var}(\tilde{Y}_0\mid X=x)}{\operatorname{Var}(\tilde{Y}_1\mid X=x)} = \text{constant}, \) or if  \(  \rho = \pm 1.
\)  
This follows because:  
\begin{itemize}
    \item Under \textit{Conditional Coverage}  we have \(\phi_t(x) \geq 0.9\) and hence $\mathbb{E}[w(X)\phi_0(X)+ (1-w(X)) \phi_1(X)] \geq  \mathbb{E}[w(X)0.9+ (1-w(X))0.9] \geq 0.9$.  
    \item  Under \textit{Similarity}, we have \(\phi_0(x) = \phi_1(x)\) and hence $\mathbb{E}[w(X)\phi_0(X)+ (1-w(X)) \phi_1(X)] = \mathbb{E}[\phi_0(X)] \geq 0.9$ .  
    \item Under \textit{Homoskedasticity} or constant variance, we have \(w(x) = \text{constant}\) and hence $\mathbb{E}[w\phi_0(X)+ (1-w) \phi_1(X)]= w \mathbb{E}[\phi_0(X)]+ (1-w)  \mathbb{E}[\phi_1(X)]  \geq 0.9$ since $\mathbb{E}\phi_t(X) \geq  0.9$. 
    \item If \(\rho = \pm 1\), then \(\frac{\operatorname{Var}(\tilde{Y}_0\mid X=x)}{\operatorname{Var}(\tilde{Y}_1\mid X=x)} = \text{constant}\).  
\end{itemize}

\textbf{Observation 2: } If $\rho =-1$ then 
\begin{align*}
& D_{-1}\big(u_0(x), u_1(x)\big) - D_{-1}\big(u_0(x)-e_0(x), u_1(x)-e_1(x)\big)\\&= u_0(x) + u_1(x) -  (u_0(x)-e_0(x))- (u_1(x)-e_1(x)) \\&= e_0 + e_1.
 \end{align*}
Hence always $e_0 + e_1 \leq  D_{-1}\big(u_0(x), u_1(x)\big) - D_{-1}\big(u_0(x)-e_0(x), u_1(x)-e_1(x)\big). $
\end{proof}

\begin{lemma}
\label{lemma_subGauss} 
Let $X, Y$ be a pair of Gaussian centered random variables with correlation $\rho$. Let $q_X, q_Y\in\mathbb{R}$ such that $$\mathbb{P}(X\leq q_X)= \phi_0, \,\,\,\,\,\, \mathbb{P}(Y\leq   q_Y)= \phi_1, \quad \phi_0, \phi_1\in (0,1).$$
Then, 
\begin{equation}\label{eq43576}
    \mathbb{P}(X+Y \leq \sqrt{q_X^2 + q_Y^2 + 2\rho q_Xq_Y })\geq (1-w)\phi_0+w\phi_1,
\end{equation}
where $w =(1-\rho^2) \frac{\sigma_Y^2}{\sigma_{X+Y}^2}\in(0,1)$. 

\end{lemma}
\begin{proof}
Denote the following:

\begin{itemize}
    \item $\phi_2 = (1-w)\phi_0+w\phi_1$. 
    \item     $\Phi^{-1}(x)$ the quantile function of standard Gaussian distribution
    \item  $h(x):=(\Phi^{-1}(x))^2$. It is a simple exercise to show that $h$ is a convex function (see the proof below).
\end{itemize}

First we show the case $\rho=0$. We use this result to prove general case $\rho\in[-1,1]$ by decomposing $X+Y$ into independent components $X+Y = Z+\varepsilon$ where $Z \indep\varepsilon$. 

\textbf{Case $\rho = 0$:} We want to show 
\begin{align*}
    \operatorname{Quantile}_{\phi_2}(X+Y)\leq \sqrt{q_X^2 + q_Y^2 }. 
\end{align*}
We can obviously consider that the equalities hold: $$\mathbb{P}(X\leq q_X)= \phi_0, \,\,\,\,\,\, \mathbb{P}(Y\leq   q_Y)= \phi_1,$$
otherwise the right side of \eqref{eq43576} is just be larger. Note that in such a case, $q_X = \Phi^{-1}(\phi_0)\sigma_X$,  $q_Y = \Phi^{-1}(\phi_1)\sigma_Y$. Finally, we have
\begin{align*}
\operatorname{Quantile}_{\phi_2}(X+Y) &= \Phi^{-1}(\phi_2)\sqrt{\sigma_X^2 + \sigma_Y^2}\\&
=\sqrt{h(\phi_2)\sigma_X^2 + h(\phi_2)\sigma_Y^2}\\&
\leq \sqrt{h(\phi_0)\sigma_X^2 + h(\phi_1)\sigma_Y^2}\\&
= \sqrt{q_X^2 + q_Y^2 },
\end{align*}
where the inequality follows directly from the convexity of $h$. 

\textbf{Case $\rho> 0$:}

Without loss of generality, assume that $\phi_1\geq\phi_0$; otherwise we can make an analogous argument exchanging $X$ and $Y$. We can obviously consider that the equalities hold: $$\mathbb{P}(X\leq q_X)= \phi_0, \,\,\,\,\,\, \mathbb{P}(Y\leq   q_Y)= \phi_1,$$
otherwise the right side of \eqref{eq43576} is just be larger.

Let us write $Y=\alpha X + \varepsilon$, where $X\indep \varepsilon$. This is possible from a basic linear regression property of jointly Gaussian random variables.  We also define $Z = (1+\alpha)X$. Notice the following properties: $$\alpha = \rho\frac{\sigma_Y}{\sigma_X}, \,\,\,\sigma_\varepsilon^2 = \sigma_Y^2 - \alpha^2\sigma_X^2, \,\,\,\sigma_Z^2 = (1+\alpha)^2\sigma_X^2, \,\,\,q_X = \Phi^{-1}(\phi_0)\sigma_X, \,\,\,q_Y = \Phi^{-1}(\phi_1)\sigma_Y.$$

We want to show 
\begin{align*}
    \operatorname{Quantile}_{\phi_2}(X+Y)\leq \sqrt{q_X^2 + q_Y^2+ 2\rho q_Xq_Y }. 
\end{align*}
However, 
\begin{align*}
    \operatorname{Quantile}_{\phi_2}(X+Y) =    \operatorname{Quantile}_{\phi_2}(Z+\varepsilon) \leq \sqrt{\sigma_Z^2h(\phi_0) +\sigma_\varepsilon^2h(\phi_1) } 
\end{align*}
from the already solved independent case. 

It is a simple algebraic exercise to show that $\sqrt{\sigma_Z^2h(\phi_0) +\sigma_\varepsilon^2h(\phi_1) } \leq  \sqrt{q_X^2 + q_Y^2+ 2\rho q_Xq_Y } $: Lets rewrite

\begin{align*}
&\sigma_Z^2h(\phi_0) +\sigma_\varepsilon^2h(\phi_1) =   
 (1+\alpha)^2\sigma_X^2h(\phi_0) + (\sigma_Y^2 - \alpha^2\sigma_X^2)h(\phi_1), \\&
 q_X^2 + q_Y^2+ 2\rho q_Xq_Y = h(\phi_0)\sigma_X^2 + h(\phi_1)\sigma_Y^2+ 2\rho \sigma_X\sigma_Y\sqrt{h(\phi_0)h(\phi_1)}
\end{align*}
and hence 
\begin{align*}
&\sqrt{\sigma_Z^2h(\phi_0) +\sigma_\varepsilon^2h(\phi_1) } \leq  \sqrt{q_X^2 + q_Y^2+ 2\rho q_Xq_Y } \\& \iff\\&
 (1+\alpha)^2\sigma_X^2h(\phi_0) + (\sigma_Y^2 - \alpha^2\sigma_X^2)h(\phi_1) \leq h(\phi_0)\sigma_X^2 + h(\phi_1)\sigma_Y^2+ 2(\alpha \frac{\sigma_X}{\sigma_Y}) \sigma_X\sigma_Y\sqrt{h(\phi_0)h(\phi_1)}\\&\iff\\&
 \alpha^2\bigg( h(\phi_1) -   h(\phi_0)\bigg) +2 \alpha \bigg( \sqrt{h(\phi_0)h(\phi_1)} - h(\phi_0) \bigg)\geq 0. 
\end{align*}
Notice that since we postulated $h(\phi_1)\geq h(\phi_0)$, this inequality is always satisfied for $\alpha\geq 0$. Since $\alpha\geq 0 \iff \rho \geq 0$, we showed that  $\sqrt{\sigma_Z^2h(\phi_0) +\sigma_\varepsilon^2h(\phi_1) } \leq  \sqrt{q_X^2 + q_Y^2+ 2\rho q_Xq_Y } $ and therefore  $    \operatorname{Quantile}_{\phi_2}(X+Y)\leq \sqrt{q_X^2 + q_Y^2+ 2\rho q_Xq_Y }$, what we wanted to show.

\textbf{Case $\rho< 0$:}

Analogously to the case $\rho>0$, we write $X = -\alpha Y  + \varepsilon$ where $Y\indep \varepsilon$ and $\alpha>0$. Using the same steps, 
\begin{align*}
&    \operatorname{Quantile}_{\phi_2}(X+Y)\leq \sqrt{q_X^2 + q_Y^2+ 2\rho q_Xq_Y }\\&\iff\\&
 \alpha^2\bigg(  h(\phi_0) -  h(\phi_1) \bigg) +2 \alpha \bigg(   h(\phi_0) - \sqrt{h(\phi_0)h(\phi_1)}  \bigg)\leq 0. 
\end{align*}
Notice that since we postulated $h(\phi_1)\geq h(\phi_0)$, this inequality is always satisfied for $\alpha> 0$. Therefore, $    \operatorname{Quantile}_{\phi_2}(X+Y)\leq \sqrt{q_X^2 + q_Y^2+ 2\rho q_Xq_Y }$, what we wanted to show. 

\textbf{Proof that $h$ is convex: } Proof by picture is in Figure~\ref{convex_gauss}. Also used in \citep{bodik2025identifiabilitycausalgraphsnonadditive}. Formally, to prove that \( h(x) = \left( \Phi^{-1}(x) \right)^2 \) is convex, we need to show that the second derivative of \( h(x) \) is non-negative for all \( x \in (0, 1) \).

\begin{figure}[h!]
    \centering
    \includegraphics[width=0.5\linewidth]{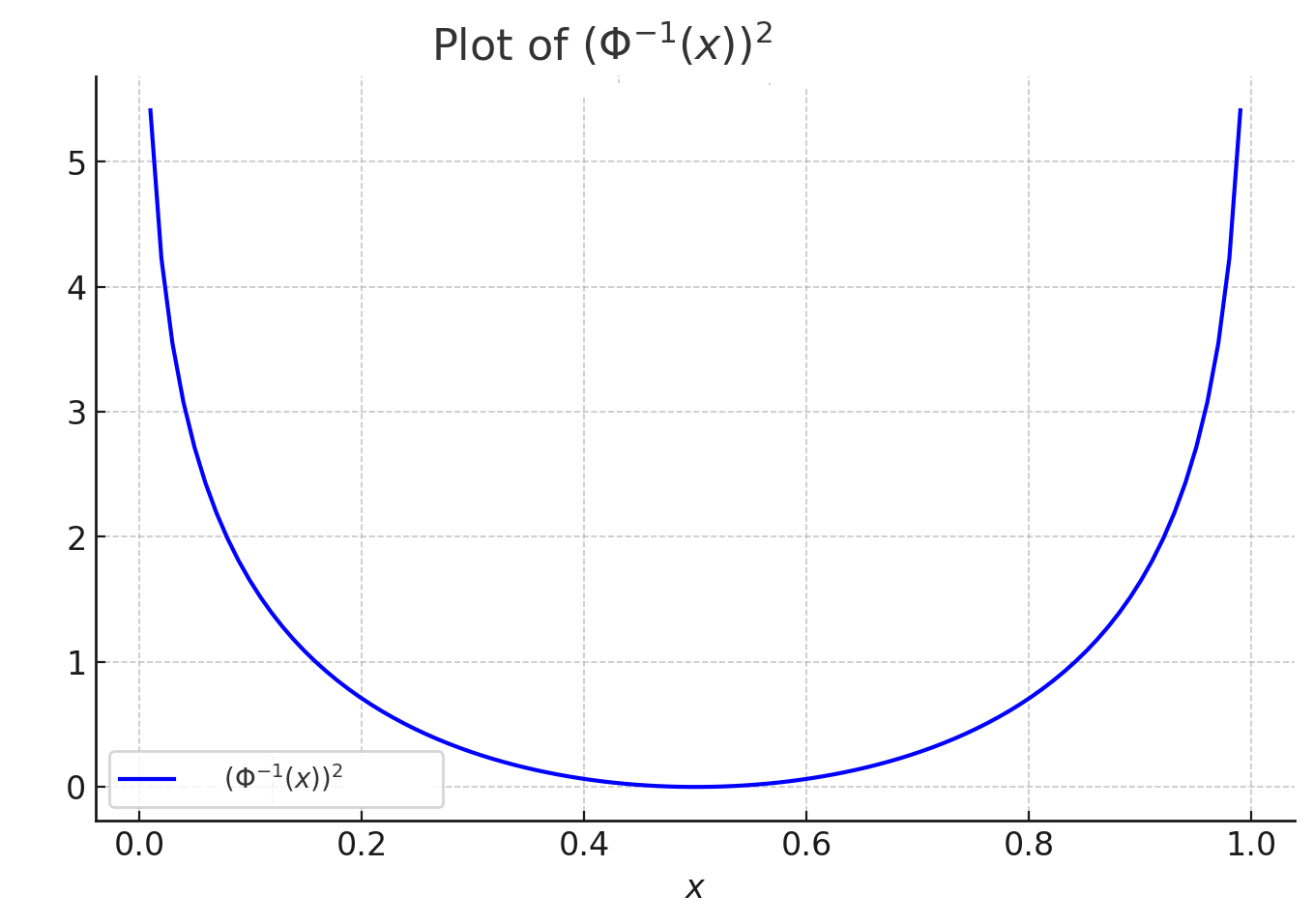}
    \caption{Convexity of  $h(x):=\big( \Phi^{-1}(x)\big)^2 $ can be seen directly from this plot without any algebraic computations. }
    \label{convex_gauss}
\end{figure}

Let \( \Phi^{-1}(x) = z \), so that \( \Phi(z) = x \). Then, \( h(x) = z^2 \). To compute the first derivative of \( h(x) \), apply the chain rule:
\[
h'(x) = 2z \cdot \frac{d}{dx} \left( \Phi^{-1}(x) \right).
\]
The derivative of the inverse function \( \Phi^{-1}(x) \) is given by:
\[
\frac{d}{dx} \left( \Phi^{-1}(x) \right) = \frac{1}{\phi(\Phi^{-1}(x))},
\]
where \( \phi(x) \) is the standard normal density function. Thus, the first derivative of \( h(x) \) is:
\[
h'(x) = 2 \Phi^{-1}(x) \cdot \frac{1}{\phi(\Phi^{-1}(x))}.
\]
Next, we compute the second derivative of \( h(x) \). We will again apply the product rule and chain rule:
\[
h''(x) = 2 \frac{d}{dx} \left( \Phi^{-1}(x) \right) \cdot \frac{1}{\phi(\Phi^{-1}(x))} + 2 \Phi^{-1}(x) \cdot \frac{d}{dx} \left( \frac{1}{\phi(\Phi^{-1}(x))} \right).
\]
Now, differentiate \( \frac{1}{\phi(\Phi^{-1}(x))} \):
\[
\frac{d}{dx} \left( \frac{1}{\phi(\Phi^{-1}(x))} \right) = - \frac{\phi'(\Phi^{-1}(x))}{\phi(\Phi^{-1}(x))^2} \cdot \frac{d}{dx} \left( \Phi^{-1}(x) \right).
\]
Using the fact that \( \phi'(z) = -z \phi(z) \), we get:
\[
\frac{d}{dx} \left( \frac{1}{\phi(\Phi^{-1}(x))} \right) = \frac{z}{\phi(z)^2} \cdot \frac{1}{\phi(z)} = \frac{z}{\phi(z)^3}.
\]
Thus, the second derivative of \( h(x) \) becomes:
\[
h''(x) = 2 \cdot \frac{1}{\phi(\Phi^{-1}(x))^2} + 2 \Phi^{-1}(x) \cdot \frac{\Phi^{-1}(x)}{\phi(\Phi^{-1}(x))^3}.
\]
We now simplify the expression for \( h''(x) \):
\[
h''(x) = \frac{2}{\phi(\Phi^{-1}(x))^2} + \frac{2 (\Phi^{-1}(x))^2}{\phi(\Phi^{-1}(x))^3}.
\]
Since \( \phi(z) > 0 \) for all \( z \), both terms are positive, implying that \( h''(x) > 0 \) for all \( x \in (0, 1) \). Therefore, the function \( h(x) \) is convex on \( (0, 1) \).
\end{proof}

\begin{customthm}{\ref{theorem_CI}}[Conditional validity of \(CW^{+CI}(\rho) \) intervals when \( \rho = 1 \)]
\,
\begin{tcolorbox}[colback=gray!5!white, colframe=gray!20!white, boxrule=0.2pt, left=1pt, right=1pt]
\begin{flushleft}
For notational convenience in the proof, we re-define 
\[
\tilde{Y}_0 := -Y(0), \qquad \tilde{Y}_1 := Y(1), \quad \text{ and } \quad f_t(x) := \mathbb{E}[\tilde{Y}_t \mid X=x], \quad t=0,1,
\]
so that we are interested in $ITE = \tilde{Y}_1 + \tilde{Y}_0$ and CATE becomes $\tau(x) = f_0(x) + f_1(x)$. With this notation in place, an equivalent formulation of Theorem~\ref{theorem_CI} is given below. 
\end{flushleft}
\end{tcolorbox}

Let \( x \in \mathcal{X} \) and assume \( \rho(x) = 1 \), i.e., \( \operatorname{cor}(\tilde{Y}_0, \tilde{Y}_1 \mid X = x) = -1 \). Suppose that for some \( \beta \in (0,1) \), our confidence intervals satisfy \( \mathbb{P}\big(\tau(x) \leq \hat{\tau}(x) + r_u(x)\big) \geq 1 - \beta \). For any $\tilde\rho\in[-1,1]$ holds:

\begin{itemize}
    \item If \( \operatorname{var}(\tilde{Y}_0 \mid X = x) = \operatorname{var}(\tilde{Y}_1 \mid X = x) \), then:
\[
\mathbb{P}\big(\tilde{Y}_0 + \tilde{Y}_1 \leq \hat{\tau}(X) + r_u(X) + D_{\tilde\rho}(u_0(X), u_1(X)) \mid X = x\big) \geq 1 - \beta.
\]
\item  If \( \operatorname{var}(\tilde{Y}_0 \mid X = x) \neq \operatorname{var}(\tilde{Y}_1 \mid X = x) \), and the following conditions hold for each \(t \in \{0,1\}\):
    \begin{align*}
    \mathbb{P}(\tilde{Y}_t \leq \hat{f}_t(X) + u_t(X) \mid X = x) &= 1 - \alpha, \label{eq:cond_validity} \\
    \mathbb{P}\left(f_t(x) \leq \hat{f}_t(x) + \tfrac{1}{2} r_u(x)\right) &\geq 1 - \tfrac{1}{2} \beta,
    \end{align*}
    then the \(CW^{+CI}(\rho)\) prediction intervals are conditionally valid:     \[
    \mathbb{P}\left(\tilde{Y}_0 + \tilde{Y}_1 \leq \hat{\tau}(X) + r_u(X) + D_{\tilde{\rho}}(u_0(X), u_1(X)) \mid X = x\right) \geq (1 - \alpha)(1 - \beta).
    \]
\end{itemize}
\end{customthm}

\begin{proof}
As in the proof of Theorem~\ref{rho_1_theorem}, we use the following classical result: for any two random variables \( Z_0 \) and \( Z_1 \) with finite correlation, we have
\[
\operatorname{cor}(Z_0, Z_1) = -1 \quad \Leftrightarrow \quad Z_1 = a - b Z_0
\]
for some constants \( a, b \in \mathbb{R} \) with \( b > 0 \). This follows directly from the Cauchy-Schwarz inequality.

Conditioning on \( X = x \), we apply this result to the pair \( (\tilde{Y}_0, \tilde{Y}_1) \), and define the relation
\[
\tilde{Y}_1 = a_x - b_x \tilde{Y}_0,
\]
(equality is always meant as almost sure equality), where
\[
b_x = \frac{\operatorname{sd}(\tilde{Y}_1 \mid X = x)}{\operatorname{sd}(\tilde{Y}_0 \mid X = x)} = \sqrt{ \frac{\operatorname{var}(\tilde{Y}_1 \mid X = x)}{\operatorname{var}(\tilde{Y}_0 \mid X = x)}}.
\]
Let \( \tilde{Y}^{cen}_t := \tilde{Y}_t - \mathbb{E}[\tilde{Y}_t \mid X = x] \) denote the centered versions of \( \tilde{Y}_0 \) and \( \tilde{Y}_1 \). Then
\[
\tilde{Y}^{cen}_1 = -b_x \tilde{Y}^{cen}_0.
\]

\paragraph{Case 1: \( b_x = 1 \).} We consider the following probability:
\begin{align*}
&\mathbb{P}\left(\tilde{Y}_0 + \tilde{Y}_1 \leq \hat{\tau}(X) + r_u(X) + D_{\tilde\rho}\big(u_0(X), u_1(X)\big) \mid X = x\right) \\
&= \mathbb{P}\left(\tau(X) \leq \hat{\tau}(X) + r_u(X) + D_{\tilde\rho}\big(u_0(X), u_1(X)\big) \mid X = x\right) \\
&\geq \mathbb{P}\left(\tau(X) \leq \hat{\tau}(X) + r_u(X) \mid X = x\right) \geq 1 - \beta.
\end{align*}

\paragraph{Case 2: \( b_x < 1 \).} Using the identity \( \tilde{Y}^{cen}_1 = -b_x \tilde{Y}^{cen}_0 \), we write:
\[
\tilde{Y}_0 + \tilde{Y}_1 = \tilde{Y}^{cen}_0 + \tilde{Y}^{cen}_1 + f_0(x) + f_1(x) = (1 - b_x)\tilde{Y}^{cen}_0 + f_0(x) + f_1(x).
\]
So we have:
\begin{align*}
&\mathbb{P}\left(\tilde{Y}_0 + \tilde{Y}_1 \leq \hat{\tau}(x) + r_u(x) + D_{\tilde\rho}(u_0(x), u_1(x)) \mid X = x\right) \\
&= \mathbb{P}\left((1 - b_x)\tilde{Y}^{cen}_0 + f_0(x) + f_1(x) \leq \hat{f}_0(x) + \hat{f}_1(x) + r_u(x) + D_{\tilde\rho}(u_0(x), u_1(x)) \mid X = x\right) \\
&= \mathbb{P}\left((1 - b_x)\tilde{Y}^{cen}_0 \leq e_0(x) + e_1(x) + r_u(x) + D_{\tilde\rho}(u_0(x), u_1(x)) \mid X = x\right),
\end{align*}
where \( e_t(x) = \hat{f}_t(x) - f_t(x) \) for \( t = 0, 1 \).

We are given that:
\[
\mathbb{P}\left(\tilde{Y}^{cen}_0 \leq e_0(x) + u_0(x) \mid X = x\right) = 0.9, \quad \mathbb{P}\left(b_x \tilde{Y}^{cen}_0 \leq e_1(x) + u_1(x) \mid X = x\right) = 0.9.
\]
From this, it follows that:
\[
e_0(x) + u_0(x) = \frac{e_1(x) + u_1(x)}{b_x}.
\]
Therefore:
\[
\mathbb{P}\left((1 - b_x)\tilde{Y}^{cen}_0 \leq |(e_0(x) + u_0(x)) - (e_1(x) + u_1(x))| \mid X = x\right) = 0.9.
\]
Now, due to the assumption of \( \mathbb{P}(f_t(x) \leq \hat{f}_t(x) + \tfrac{1}{2}r_u(x))\geq 1-\beta/2 \), we have that
\begin{equation*}
    \mathbb{P}\bigg(r_u(x) \geq -2 \max\{e_0(x), e_1(x)\} \bigg)\geq 1-\beta.
\end{equation*}
Note the following fact:
\[
\begin{aligned}
r_u(x) &\geq -2 \max\{e_0(x), e_1(x)\} \\
\implies \quad |(e_0(x) + u_0(x)) - (e_1(x) + u_1(x))| &\leq e_0(x) + e_1(x) + r_u(x) + |u_0(x) - u_1(x)|.
\end{aligned}
\]
Hence, with probability at least \( 1 - \beta \), we have:
\[
|(e_0(x) + u_0(x)) - (e_1(x) + u_1(x))|  \leq e_0(x) + e_1(x) + r_u(x) + |u_0(x) - u_1(x)|.
\]
Putting it all together:
\[
\mathbb{P}\left(\tilde{Y}_0 + \tilde{Y}_1 \leq C_{ITE}(x) \mid X = x\right) \geq (1-\alpha)(1-\beta),
\]
where $C_{ITE}(x) = \hat{\tau}(x) + r_u(x) + D_\rho(u_0(x), u_1(x))$. 

An analogous argument can be done when $b_x>1$ as the problem is symmetric in $\tilde{Y}_0\leftrightarrow \tilde{Y}_1$. 
\end{proof}



\end{document}